\documentclass[11pt]{amsart}
\usepackage[normalem]{ulem}
\usepackage{amssymb}
\usepackage{amsthm}
\usepackage{amsmath}
\usepackage{geometry}                
\usepackage[linesnumbered,ruled,vlined, commentsnumbered]{algorithm2e}

\usepackage{placeins}
\usepackage{graphicx}

\usepackage{color}
\usepackage[colorlinks=true]{hyperref}
\usepackage{cleveref}
\usepackage{subcaption}
\usepackage{mathrsfs}
\usepackage{comment}
\usepackage[skip=2pt,font=small]{caption}
\usepackage{amsmath}
\usepackage[shortlabels]{enumitem}  
\usepackage{placeins}
\usepackage{accents}
\usepackage{bm}
\usepackage{bbm}

\usepackage{mathtools,xparse}
\usepackage{amssymb}

\DeclarePairedDelimiter{\abs}{\lvert}{\rvert}

\def\EE{{ \mathbb{E}}}
\def\RR{{\mathbb R}}
\def\NN{{ \mathbb{N}}}

\def\vel{{\theta}}
\def\velvec{\vec{\vel}}

\def\dd{{\rm d}}

\def\xprop{Q}
\def\xpropsc{\mathcal{\xprop}}
\def\xtrans{P}
\def\xtarget{\pi}

\def\prop{\widetilde{Q}}

\def\trans{\widetilde{P}}
\def\target{\widetilde{\pi}}

\def\xtransc{\mathcal{\xtrans}}

\def\pgraph{G_{p}}
\def\pvertices{V_{p}}
\def\pedges{E_{p}}

\def\dgraph{G_{d}}
\def\dvertices{V_{d}}
\def\dedges{E_{d}}

\def\ndist{n_{D}}
\def\conflict{C}
\def\embedding{\phi}

\def\graph{\mathcal{G}}

\def\vertices{\mathcal{V}}
\def\edges{\mathcal{E}}

\def\nb{\mathscr{N}}

\def\vectorfield{\mathbf{v}}
\def\centroid{\mathbf{c}}

\def\nb{\mathcal{N}}

\def\activeset{\mathcal{A}}

\def\EE{{ \mathbb{E}}}

\def\Domain{{\dmDomain}}

\def\dm{ \xi}

\def\dmDomain{ \Xi}

\def\x{x}
\def\xDomain{\mathcal{X}}

\def\vf{\mathbf{v}}
\def\orientation{\text{orientation}_{\embedding,\vf}}

\def\flip{F}

\def\xweight{\omega}
\def\weight{\widetilde{\omega}}

\def\cedges{C}

\DeclareMathOperator\supp{supp}

\newcommand{\ccdot}{\,\cdot\,}
\newtheorem{proposition}{Proposition}[section]

\newtheorem{lemma}[proposition]{Lemma}
\newtheorem{assumption}{Assumption}
\newtheorem{theorem}[proposition]{Theorem} 
\crefname{algorithm}{Algorithm}{Algorithms}
\crefname{assumption}{Assumption}{Assumptions}
\crefname{lemma}{Lemma}{Lemmas}
\crefname{theorem}{Theorem}{Theorems}
\crefname{remark}{Remark}{Remarks}
\crefname{corollary}{Corollary}{Corollaries}
\crefname{figure}{Fig.}{Figures}
\crefname{section}{Section}{Sections}

\theoremstyle{remark}
\newtheorem{remark}[proposition]{Remark}

\DeclareMathOperator{\sgn}{sgn}
\newcommand{\gh}[1]{\textcolor{blue}{GH: #1}}

\newcommand{\pop}{\text{pop}}
\newcommand{\area}{\text{area}}
\newcommand{\boundary}{\text{boundary}}

\makeatletter
\g@addto@macro{\endabstract}{\@setabstract}
\newcommand{\authorfootnotes}{\renewcommand\thefootnote{\@fnsymbol\c@footnote}}%

\def\dmwidth{.47}
\title{Non-reversible Markov chain Monte Carlo for sampling of districting maps }

\begin{document}
\begin{center}
  \LARGE 
  Non-reversible Markov chain Monte Carlo for sampling of districting maps\\  \par \bigskip

  \normalsize
  \authorfootnotes
  Gregory Herschlag\footnote{gjh@math.duke.edu}\textsuperscript{1}, Jonathan C. Mattingly\footnote{jonm@math.duke.edu}\textsuperscript{1,2},
  Matthias Sachs\footnote{msachs@math.duke.edu}\textsuperscript{1,3}, Evan Wyse\footnote{evan.wyse@duke.edu}\textsuperscript{2} \par \bigskip

  \textsuperscript{1}Department of Mathematics, Duke University, Durham, NC 27708 \par
  \textsuperscript{2}Department of Statistical Science, Duke University, Durham NC 27708\par 
  \textsuperscript{3}The Statistical and Applied Mathematical Sciences Institute (SAMSI), Durham, NC 27709\par
  
  \bigskip

\end{center}

\begin{abstract}
%
  Evaluating the degree of partisan districting (Gerrymandering) in a statistical framework typically requires an ensemble of districting plans which are drawn from a prescribed probability distribution that adheres to a realistic and non-partisan criteria. 
  In this article we introduce novel non-reversible Markov chain Monte-Carlo (MCMC) methods for the sampling of such districting plans which have improved mixing properties in comparison to previously used (reversible) MCMC algorithms. In doing so we extend the current framework for construction of non-reversible Markov chains on discrete sampling spaces by considering a generalization of skew detailed balance. We provide a detailed description of the proposed algorithms and evaluate their performance in numerical experiments.
\end{abstract}
\section{Introduction}
The use of computer generated alternative redistricting plans to benchmark particular redistricting maps has gained legal and scientific traction in recent years. The generation of such an \textit{ensemble} of maps has been used to identify and quantify the extent of partisan and racial gerrymandering by answering the question ``What would one expect to have happened if no partisan or racial information had been used?'' These methods produce a baseline informed by the geo-political geography of the state and  which do not assume proportional presentation or unrealistic symmetry assumptions. This baseline can then be used to evaluate a particular redistricting plan of interest.

In \cite{mattingly2014redistricting,fifield2015new,jcmReport,najt2019complexity}, an ensemble of maps is generated by sampling from probability distribution constructed on the space of possible redistricting plans using only non-partisan considerations. In this thread of work, the sampling was performed via Markov chain Monte Carlo (MCMC) using a standard Metropolis-Hasting algorithm based on a single node flip proposal chain. Other ensemble methods have used generative techniques based on optimization, genetic algorithms, or Markov chains without a clearly
describable stationary measure. Examples of the latter include the generation of samples using simulated annealing \cite{bangia2017redistricting,herschlag2017evaluating,herschlag2018quantifying} and Markov chains based on merge-split operations \cite{deford2019recombination} (see also \cite{carter2019merge} for an extension of the latter work which allows one to generate samples from a prescribed target measure.)

Many of the above samplers utilize the Metropolis-Hastings algorithm so the underlying generating Markov chain is {\em reversible} The reversible methods, by definition, have a Markov kernel associated with this Markov chain satisfying a detailed balance condition with respect to the corresponding stationary measure. Heuristically, this implies that the Markov chain has a diffusive nature. Other samplers used which are non-reversible  typically sample from an unknown distribution.

In recent years MCMC methods based on non-reversible Markov chains (i.e., Markov chains whose Markov kernel fails to satisfy a detailed balance condition) have attracted increased attention because of their favorable convergence and mixing properties; without claim to completeness of the work listed we refer the reader to \cite{neal2004improving,sun2010improving,turitsyn2011irreversible,suwa2012general,hukushima2013irreversible} and to \cite{bierkens2019zig,michel16:_irrevMCMC,bouchard2018bouncy,duncan2016variance,duncan2017nonreversible,ma2019irreversible,dobson2020reversible} for examples of non-reversible MCMC methods for sampling on continuous spaces, and discrete spaces, respectively.

For many of these methods, improved mixing properties over their reversible counterparts is folklore among practitioners; however, there is a growing body of theoretical work that supports these claims \cite{diaconis2000analysis,duncan2016variance}.  

In the setup of a continuous sampling space, non-reversible Markov chains naturally arise through time discretization of stochastically perturbed versions or modifications of Newton's equations of motion (see e.g., \cite{akhmatskaya2009comparison,ottobre2016function}). In these cases, reversibility of the dynamics is broken due to the presence of inertia modeled by the momenta associated with each degree of freedom. This is consistent with physical intuition that the resulting ballistic-like motion tends to exhibit better mixing properties over a purely diffusive dynamics of a reversible Markov chain. For example, the existence of momentum is typically cited as the strength of Langevin sampling over simple Browning dynamics.





For sampling in discrete space, a common approach for designing non-reversible MCMC methods is what is sometimes referred to as ``lifting'' \cite{vucelja2016lifting,michel16:_irrevMCMC}. Here, a reversible MCMC method is modified by replicating the state space through the introduction of a dichotomous auxiliary variable taking values in $\{-1,1\}$ along with a simple directed subgraph of the Markov state graph induced by the original reversible Markov chain. Depending on the value of the auxiliary variable transition probabilities along the assigned directions of the simple subgraph are then either increased or decreased. As such the auxiliary variable has a similar effect as the momentum variable in the continuous setup. For example, in the case where the  Markov state graph induced by the original reversible MCMC method is a circular graph, a simple way of implementing a lifting approach is by increasing clockwise transition probabilities for a positive value of the auxiliary variable and increasing counter-clockwise transition probabilities for a negative values of the auxiliary variable \cite{diaconis2000analysis}.


The sampling efficiency of the Markov chain obtained by lifting highly depends on the choice of the simple directed subgraph. While in the above mentioned example a ``good'' choice can easily be identified, constructing a suitable subgraph for Markov chains whose associated Markov state graph has a more complex topology can be difficult.



In this article, we introduce non-reversible MCMC methods for the sampling of redistricting maps.
Creating a collection of redistricting maps, via sampling of a specified measure, is an important step in many method currently used to evaluate redistricting and detect and explain gerrymandering.
In this note, we introduce a heuristic for implementing an efficient lifting approach which is based on a notion of flowing the center of mass of districts along a defined vector field; the center of mass arises from an embedding of the districting graph in $\RR^{2}$.
We also introduce a novel construction for non-reversible MCMC dynamics as a generalization of the standard lifting approach which allows the incorporation of multiple momenta variables. This allows us to construct non-reversible MCMC schemes for our application which make use of the structure of the induced district-level graph. Finally, we combine these methods with a tempering scheme which minimizes rejection rates in the non-reversible Markov chain and thereby increases sampling efficiency.

The remainder of this article is organized as follows. In \cref{sec:non-rev:MH}, we review the formal definition of non-reversibility of Markov chains on discrete sampling spaces. In \cref{sec:skewed:db}, we review the basic construction of non-reversible MCMC methods via a skew detailed balance condition. In \cref{sec:general}, we describe a novel construction of non-reversible MCMC schemes which allows for multipule momentum corresponding to different proposal chains. In \cref{sec:application}, we describe the implementation of our approach under the application under consideration; in \cref{sec:num}, we test our ideas numerically.

\section{Exposition of the main algorithms}
Before we rigorously develop the underlying mathematical framework, we start by informally describing the two main sampling algorithms proposed in this article and demonstrating how they would be applied to sampling redistricting plans for the North Carolina Congressional Delegation. We construct our non-reversible sampling methods as modifications of a variant of the single node flip algorithms (see \cref{sec:SingleNodeFlip}) where random redistricting maps are sequentially generated by changing in each iteration the color (district allocation) of a single precinct located at the border of the current redistricting map. We introduce non-reversibility by directing transitions along what we informally refer to as a flow. Depending on the value of a momentum variable, only transitions in positive or negative direction are permitted, resulting in a macroscopic level kinetic like movements along/against the flow. 
\FloatBarrier
The intuition behind the first proposed method (``Center-of-mass flow'', see \cref{sec:precinct:flow} and \cref{sfig:cm-flow} ) is that a fast mixing on a macroscopic level is obtained if districts tend to collectively follow the flow of a suitable, well-stirring vector field in $\RR^{2}$ where there district graph is embedded. For example, under an appropriate choice of the vector field, the  resulting collective rotational movements of districts in the course of a simulation produces more efficient mixing then more diffusive sampling algorithms (see \cref{sfig:cm-flow}). Technically, we implement this idea by aligning the movement of the center of masses of districts with the vector field. For positive/negative momentum value only transitions for which the midpoints of the center of masses of the modified districts move in the positive/negative direction of the vector field are permitted.
\begin{figure} 
\includegraphics[width=0.65\linewidth, clip = true]{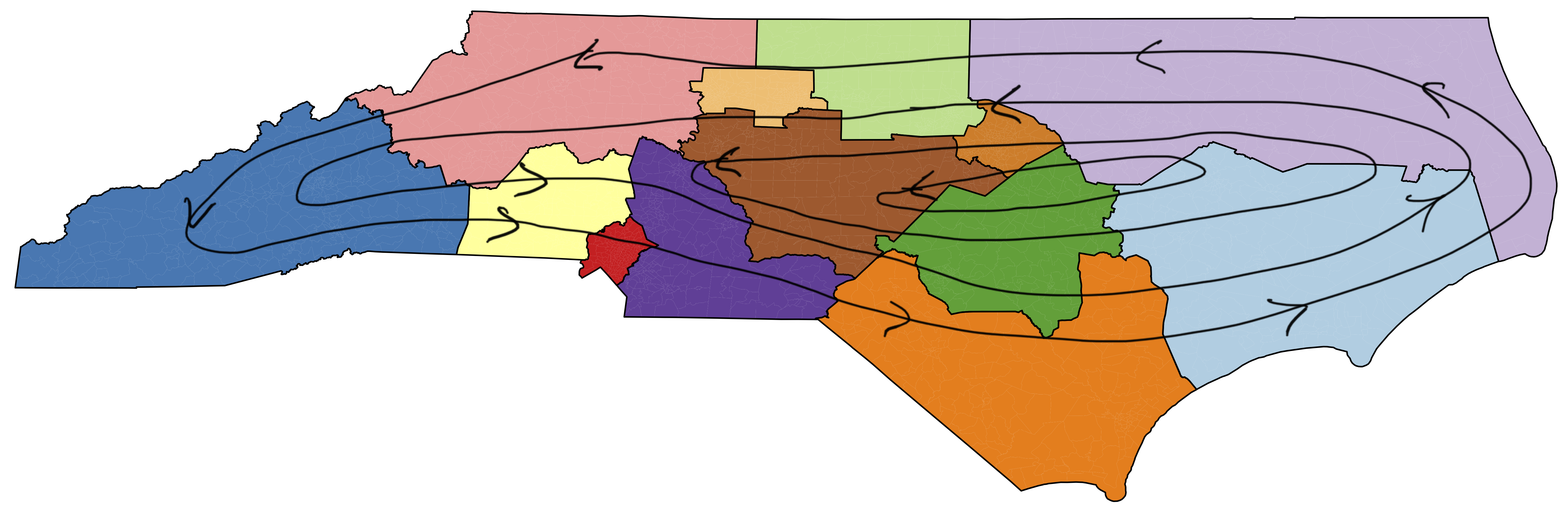}
\caption{Center-of-mass flow introduced in \cref{sec:precinct:flow}.  Changes in district boundaries must, on average, move the center of masses of the districts either with or against the drawn vector field. Drawing the  13-district N.C. Congressional map is used as an example.}\label{sfig:cm-flow}
\end{figure}
The second proposed method (``(Pair-wise) District-to-district flow'', see \cref{sec:Dist2Dist} and \cref{sfig:d2d-flow}) utilizes an extended framework, which allows the incorporation of multiple momenta each associated with a different flow. The idea of the method is to associate a momentum variable with each district pair. Depending on the value of the respective momentum, only transitions that flow districts in a direction aligned with the orientation of the respective momentum arrow are permitted. For example, consider the redistricting plan depicted in \cref{sfig:d2d-flow}. If the value of the momentum variable associated with the orange and light blue district is positive, then among the transitions which modify both these two districts only transitions that add a precinct from the light blue district to the orange district are permitted.
\begin{figure} 
\includegraphics[width=0.65\linewidth, clip = true]{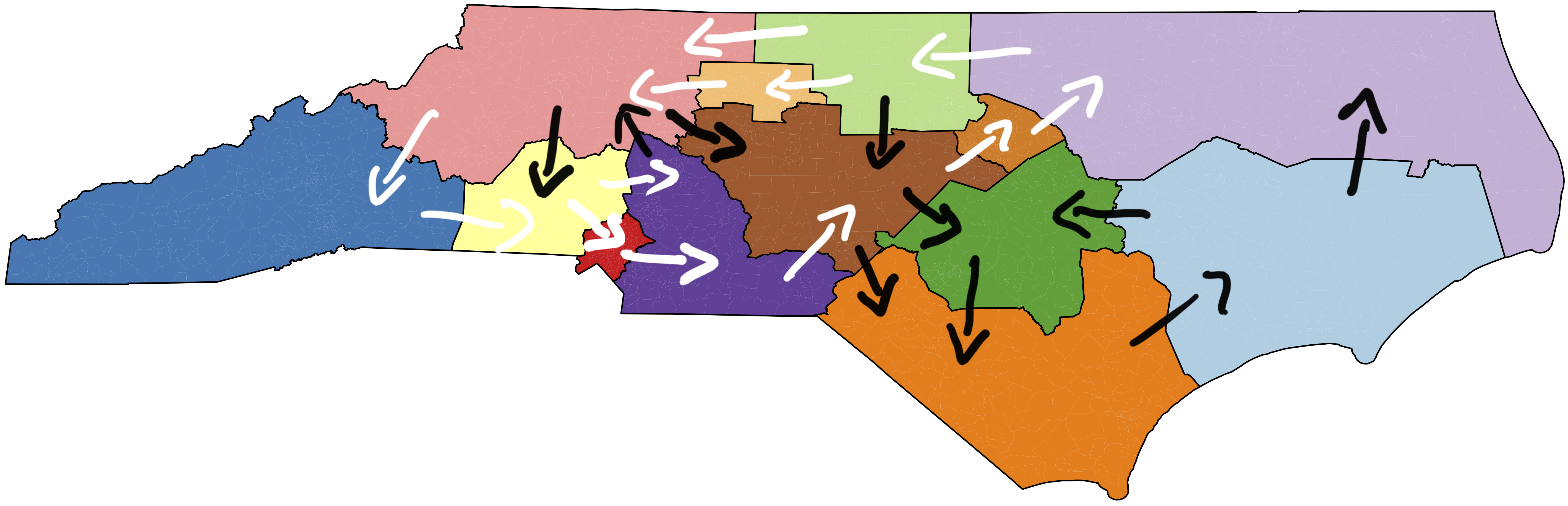}
\caption{District-to-district flow introduced in \cref{sec:Dist2Dist}.  Depending on the individual velocity values associated with each vector the corresponding boundaries between adjacent districts may either only move in the direction or in opposite direction of the displayed vector.
 }\label{sfig:d2d-flow}
\end{figure}

\FloatBarrier

\section{Reversible and Non-Reversible Markov Chains}
\subsection{Detailed Balance}\label{sec:non-rev:MH}
Consider a Markov Chain on a countable state-space $\xDomain$ with a
transition kernel $\xtransc$\footnote{We restrict to a countable state-space of simplicity. There are no inherent obstructions to generalizing to general Polish Space. See Remark~\ref{rem:genStateSpace}}.  The Markov kernel $\xtransc$ is said to be {\em reversible}, if there exists a probability distribution $\xtarget$ on
$\xDomain$ so that the pair  $(\xtransc,\xtarget)$ satisfies \emph{detailed
  balance}. That is 
\begin{align}\label{ed:detailed:balance}
 \xtarget(\x) \xtransc(\x, \x') = \xtarget(\x') \xtransc(\x', \x), \forall\x, \x'\in
  \xDomain\,. 
\end{align}
Markov transition kernels $\xtransc$ which fail to satisfy the detailed balance condition
for any measure $\xtarget$ are referred to as {\em non-reversible}.\footnote{The detailed balance condition is equivalent to the Kolmogorov definition of reversibility which requires the probability of following any sequence of states is the same as following the sequence in reverse order. This justifies the name \textit{reversible}. }
Since $\xtarget(\x) \xtransc(\x, \x')$ is the probability flux in equilibrium flowing from state $x$ to $x'$, detailed
  balance can be restated as the equilibrium flux from $x$ to $x'$ is the same as from $x'$ to $x$.

The detailed balance condition is a sufficient, 
but not a necessary
condition, for the transition kernel $\xtransc$ to preserve the measure
$\xtarget$. By definition, invariance of the measure $\xtarget$ only requires that $\xtarget \xtransc=
\xtarget$ which is just a compact notation for
\begin{align*}
   \sum_{\x \in \xDomain} \xtarget(\x) \xtransc(\x, \x') = \xtarget (\x')  \quad \forall \,\x' \in \xDomain.
\end{align*}
It can be rewritten as 
\begin{align}\label{eq:globalFlux}
  \sum_{\x' \in \xDomain\backslash\{\x\}} \xtarget(\x') \xtransc(\x', \x) 
   &= \sum_{\x' \in \xDomain\backslash\{\x\}} \xtarget(\x) \xtransc(\x, \x') \quad \forall \,\x \in \xDomain, 
\end{align}
and, as such, states that for any state $\x \in \xDomain$ the total probability flux into the state $\x$ (the lefthand side of \eqref{eq:globalFlux}) is equal to the total probability flux out of the state $\x$ (the righthand side of \eqref{eq:globalFlux}). This condition is commonly referred to as a {\em global balance condition} and is satisfied by any Markov kernel which preserves $\xtarget$.

\subsection{Skew detailed balance}\label{sec:skewed:db}

A common way of constructing non-reversible Markov chains with prescribed invariant measure $\xtarget$ is by enforcing global balance through an involutive transform.  This structure is called skew detailed balance and ensures that detailed balance holds up to some $\xtarget$-invariant involutive transformation.  More precisely,  let $S : \xDomain \rightarrow \xDomain$ be an $\xtarget$-invariant involutive transformation, so that $S=S^{-1}$, and $\xtarget(S(A)) = \xtarget(S^{-1}(A)) = \xtarget(A), ~ \forall A \subset \xDomain$. Then, the Markov kernel $\xtransc$ satisfies {\em skew detailed balance} if 
\begin{align}\label{eq:skewed:detailed:balance}
 \xtarget(\x) \xtransc(\x, \x') = \xtarget(\x') \xtransc(S(\x'), S(\x)), \forall\x, \x'\in
  \xDomain\,. 
\end{align}
It is easy to verify that skew detailed balance implies global balance (see e.g. \cite{stoltz2010free}, or proof of \cref{thm:MSDB} in \cref{sec:proofs}), and thus invariance of $\xtarget$ under $\xtransc$. 

Due to its local nature, skew detailed balance with respect to $\xtarget$ can be easily enforced by an accept-reject step. 
More precisely, let $\xprop$ denote a ``proposal'' Markov kernel on $\xDomain$, and denote by $(\x_{k})_{k\in \NN}$\footnote{Here, and in the remainder of this article we denote by $\NN$ the set of non-negative integers.} the Markov chain generated by the following generalization of the Metropolis-Hastings algorithm
\begin{enumerate}
\item $\x^{\prime} \sim \xprop(\x_{k},\ccdot)$,
\item with probability $r(\x_{k},\x^{\prime}) $ set $\x_{k+1}=\x^{\prime}$; otherwise $\x_{k+1} = S(\x_{k})$,
where 
\[
r(\x,\x^{\prime}) := \min \left ( 1, \frac{\xtarget(\x^{\prime}) \xprop(S(\x^{\prime}), \,S(\x)) }{\xtarget(\x) \xprop(\x, \,\x^{\prime})} \right ). 
\]
\end{enumerate}

Provided that the acceptance probability $r(\x,\x^{\prime}) $ is well defined for all pairs $(\x,\x^{\prime})$, the transition kernel of the generated Markov chain takes the form
\[
\xtransc(\x,\x^{\prime}) = r(\x,\x^{\prime})\xprop(\x,\x^{\prime}) + (1-r (\x,\x^{\prime}) ) \mathbbm{1}_{\{S(\x)\}} (\x^{\prime}),
\]
which indeed can be verified to satisfy the skew detailed balance condition \eqref{eq:skewed:detailed:balance} (see e.g., \cite{stoltz2010free} or proof of \cref{thm:MSMH} in \cref{sec:proofs}).

\section{A General Non-Reversible Process Construction}\label{sec:general}

In this section we first introduce a generalization of the standard skewed balance condition, termed {\em mixed skewed balance} condition, and show that this condition is sufficient for the corresponding Markov kernel to preserve a prescribed probability measure. We then provide a generalization of the Metropolis-Hastings algorithm, the {\em Mixed Skew Metropolis-Hastings Algorithm} (MSMH) which utilizes the mixed skewed balance condition. 

\subsection{The mixed skewed balance condition} 
In the following, let $\{S_{i}\}_{i=1}^{n}$ be a collection of $\xtarget$-invariant involutions, $\{\xtrans_{i}\}_{i=1}^{n}$ a collection of Markov kernels on $\xDomain$, and 
\begin{equation*}
\xweight : \xDomain \rightarrow \Delta^{n-1}, ~\xweight(\x) = \left ( \xweight_{1}(\x),\dots, \xweight_{n}(\x) \right),
\end{equation*}
a weight vector taking values in the $n$th standard simplex $\Delta^{n-1} := \{ y \in \RR^{n}: y_{i}\geq 0, ~\sum_{i=1}^{n}y_{i}=1\}$. Since at each point $x$ the weights are non-negative and sum to one, we can build a new kernel $\xtransc$ out of the collection of Markov kernels $\{\xtrans_{i}\}_{i=1}^{n}$ by setting $\xtransc = \xweight \cdot P$, which is written more explicitly as 
\begin{equation*}
\xtransc(\x,\ccdot) = \sum_{i=1}^{n} \xweight_{i}(\x) \xtrans_{i}(\x,\ccdot), \forall \x \in \xDomain\,.
\end{equation*}
From this we see clearly that $\xtransc$ is an $\x$-dependent mixture of the kernels $\{\xtrans_{i}\}$. With this comes the interpretation that a draw from  $\xtransc$ can be realized by first picking an index $i$ according to the weights and then drawing the next state according to $\xtrans_{i}$.

We say that the Markov kernel $\xtransc$ satisfies {\em mixed skewed balance} with respect to $\xtarget$, if for all $\x,\x^{\prime} \in \xDomain$ and for all $i \in\{1,\dots,n\}$,
\begin{equation}\label{eq:generalized:db:mod}
  \begin{aligned}
\xweight_{i}(\x) \xtarget(\x) \xtrans_{i}(\x, \x^{\prime}) = \xweight_{i}(\x^{\prime})  \xtarget(\x^{\prime})   \xtrans_{i}\left(S_{i}(\x^{\prime}),S_{i}(\x) \right)\,.     
  \end{aligned}
\end{equation}
As discussed further in Remark~\ref{rem:sqewMix}, equation~\eqref{eq:generalized:db:mod} means that  $i$-th kernel $\xtrans_{i}$ satisfies the  skew-detail balance condition for the an invariant measure proportional to $\pi(x)\xweight_{i}(x)$. Yet, as the following results show, by mixing these kernels according to the weights $\xweight_{i}(x)$, one obtains a Markov which has $\pi$ as its invariant measure. Typically one has that $\pi(x)=\pi(S_{i}(\x))$ and $\xweight_{i}(x)=\xweight_{i}( S_{i}(\x))$; and hence, \eqref{eq:generalized:db:mod} can again be understood as a probability flux balancing condition; the flux from $x$ to $x'$ is equal to the flux from $S_{i}(\x')$ to $S_{i}(\x)$. (Given that we chose the $i$th kernel $\xtrans_{i}$ according to the weight $\xweight_{i}$.)
\begin{theorem}\label{thm:MSDB}
If the Markov kernel $\xtransc$ defined by the collection $\{(\xweight_{i}, \xtrans_{i}, S_{i})\}_{i=1}^{n}$ satisfies mixed detailed balance with respect to $\xtarget$, then $\xtransc \xtarget = \xtarget$, i.e., the Markov kernel $\xtransc$ preserves the probability measure $\xtarget$.
\end{theorem}
For a proof of this theorem, see Appendix~\ref{sec:proofs}.

\subsection{The Mixed Skew Metropolis-Hastings algorithm} 
\label{ssec:mixedSkewMH}
Consider a collection of $n$ Markov ``proposal kernels''  $\xprop_{i}, i=1,\dots,n,$ on the subsets $\xDomain_{i} \subset \xDomain$, $i=1,\dots,n$, respectively, which form a cover of the whole domain, i.e., $\bigcup_{i=1}^{n}\xDomain_{i}=\xDomain$. Moreover, let $S_{i} : \xDomain_{i} \rightarrow \xDomain_{i}, i=1,\dots,n$ be a collection of $\xtarget$-invariant involutions. 

In what follows we describe how the collections of proposal kernels and involutions together with a suitable state dependent weight vector $\xweight : \xDomain \rightarrow \Delta^{n-1}$ can be used to generate a Markov chain which preserves the target measure $\xtarget$.

The mixed skewed balance condition provides the appropriate framework for ``patching'' these proposals kernels together to obtain a Markov chain which samples from the target measure $\xtarget$.
 
Algorithmically, this can be implemented in a two-step algorithm (see \cref{alg:generalized:MH}). In the first step of this algorithm a proposal $\x^{\prime}$ is generated from the current state $\x$ of the Markov chain as 
\[
\x^{\prime} \sim \xprop_{i}(\x,\ccdot), ~\text{where}~ i \sim \xweight(\x).
\]
The mixed skew detailed balance condition is then enforced through an accept-reject step, where the proposal is accepted with probability
\begin{equation*}
\min \left (1, \frac{\xtarget(\x^{\prime}) \xweight_i(S_{i}(\x^{\prime})) \xprop_i(S_i(\x^{\prime}),S_i(\x)) }{\xtarget(\x) \xweight_i(\x) \xprop_i(\x,\x^{\prime})} \right),
\end{equation*}
in which case the subsequent state of the Markov chain is set to be $\x^{\prime}$, and rejected otherwise, in which case the next state of the Markov chain is set to be the $i$th involutive transformation of the current state, that is $S_{i}(\x)$. 

In order for these two steps to be well-defined and the resulting transition kernel to indeed preserve the target measure $\xtarget$ we require the weight vector $\xweight$ to satisfy  
\begin{enumerate}[label=$(\mathcal{C}_{\arabic*})$]
\item\label{it:cond:weight:1} $\xweight_{i}(\x) >0 \iff \x \in \xDomain_{i},\quad i=1,\dots, n$. 
\item\label{it:cond:weight:3}  $\xweight_{i}(\x)>0$ and $\xprop_{i}(\x,\x^{\prime})>0$  if and only if  $\xweight_{i}(\x^{\prime})>0$ and $\xprop_{i}(S_{i}(\x^{\prime}),S_{i}(\x))>0$.
\item\label{it:cond:weight:2} $\xweight_{i}(S_{i}(\ccdot)) = \xweight_{i}(\ccdot),~i=1,\dots,n$, i.e., the $i$th entry of the weight vector is invariant under the $i$th involutive transformation
\end{enumerate}
Condition \ref{it:cond:weight:1}  ensures that the effective proposal kernel $\xpropsc, ~\xpropsc(\x,\ccdot) = \sum_{i=1}^{n} \xweight_{i}(\x) \xprop_{i}(\x,\ccdot), \forall \x \in \xDomain$, is well defined, and \ref{it:cond:weight:3} ensures that  the Metropolis ratio $r_{i}(\x,\x^{\prime})$ is well defined. Invariance of the $i$th weight under the $i$th involution as stated in \ref{it:cond:weight:2} ensures that the mixed skew detailed balance condition holds for the generated Markov chain. In summary, we have
\begin{theorem}\label{thm:MSMH}
Let $\xDomain_{i}\subset \xDomain, i=1,\dots,n;$ be a cover of the $\xDomain$, and let $\xprop_{i}, i=1,\dots,n;$ be Markov kernels defined on $\xDomain_{i}, i=1,\dots,n;$ respectively. Moreover, let $S_{i}, i = 1,\dots,n;$ be a collection of $\xtarget$-invariant involutions on $\xDomain$, and $\xweight : \xDomain \rightarrow \Delta^{n-1}$ be an $\x$-dependent weight vector satisfying the conditions \ref{it:cond:weight:1} to \ref{it:cond:weight:3}. Then,
 the MSMH Markov chain generated by \cref{alg:generalized:MH} possesses $\xtarget$ as an invariant measure.
\end{theorem}
For a proof of this theorem, see Appendix~\ref{sec:proofs}.

\begin{figure}
\begin{minipage}[t]{0.45\textwidth}
\begin{algorithm}[H]
\caption{Mixed Skew Metropolis-Hastings (MSMH)}
\label{alg:generalized:MH}
\SetKwData{Left}{left}\SetKwData{This}{this}\SetKwData{Up}{up}
\SetKwFunction{Union}{Union}\SetKwFunction{FindCompress}{FindCompress}
\SetKwInOut{Input}{input}
\SetKwInOut{Output}{output}
\Input{$\x$}
sample partition $i \sim \xweight(x)$\;
generate proposal $\x^{\prime} \sim \xprop_i(\x, \ccdot)$\;
$r_{i}(\x,\x^{\prime}) \newline \hspace*{2em}\gets~\frac{\xtarget(\x^{\prime}) \xweight_i(S_{i}(\x^{\prime})) \xprop_i(S_i(\x^{\prime}),S_i(\x)) }{\xtarget(\x) \xweight_i(\x) \xprop_i(\x,\x^{\prime})}$\;
sample $u \sim \mathcal{U}([0,1])$\;
\eIf{$u < r_{i}(\x,\x^{\prime})$}{
$\x \gets \x^{\prime}$
}{
$\x \gets S_i(\x)$
}
\KwRet{x}
\BlankLine
\end{algorithm}
\end{minipage}
\hfill
\begin{minipage}[t]{0.53\textwidth}
\begin{algorithm}[H]
\caption{Mixed Skew Metropolis-Hastings on extended state space}\label{alg:generalized:MHb}
\SetKwData{Left}{left}\SetKwData{This}{this}\SetKwData{Up}{up}
\SetKwFunction{Union}{Union}\SetKwFunction{FindCompress}{FindCompress}
\SetKwInOut{Input}{input}
\SetKwInOut{Output}{output}

\Input{$\dm, \velvec$}
sample partition $i \sim \weight(\dm)$\;
sample proposal $(\dm^{\prime},\velvec^{\prime}) \sim \xprop_{i}\big((\dm, \velvec), (\ccdot, \ccdot)\big)$\;
\small{$r_{i}\big((\dm,\velvec),(\dm^{\prime}, \velvec^{\prime})\big)$} $ \newline \hspace*{4em} \gets \frac{ \weight_{i}(\dm^{\prime}) \target(\dm^{\prime}) \xprop_{i}\left((\dm^{\prime}, R_{i}(\velvec^{\prime})),(\dm, R_{i}(\velvec))\right) }{ \weight_{i}(\dm) \target(\dm) \xprop_{i}\left((\dm, \velvec),(\dm^{\prime}, \velvec^{\prime})\right)}$\;
sample $u \sim \mathcal{U}([0,1])$\;
\eIf{$u < r_{i}\big((\dm, \velvec),(\dm^{\prime}, \velvec^{\prime})\big)$}{
$(\dm,\velvec) \gets (\dm^{\prime},\velvec)$
}{
$\vel_i \gets -\vel_i$
}
\KwRet{$\dm, \velvec$}
\end{algorithm}
\end{minipage}
\hfill
\caption*{Mixed Skew Metropolis-Hastings algorithm in generic form (\cref{alg:generalized:MH}) and as obtained via augmenting  the sampling space (\cref{alg:generalized:MHb}).
}
\end{figure}
\begin{remark}\label{rem:sqewMix}
The transition kernel $\xtransc$ of the Markov chain generated by \cref{alg:generalized:MH} takes the explicit form $
\xtransc(\x,\ccdot) = \sum_{i=1}^{n} \xweight_{i}(\x) \xtrans_{i}(\x,\ccdot)$ with 
\begin{equation*}
\xtrans_{i}(\x,\x^{\prime}) = \min(1,r_{i}(\x,\x^{\prime}))\xprop_{i}(\x,\x^{\prime}) + (1-\min(1,r_{i}(\x,\x^{\prime}))) \mathbbm{1}_{\{S_{i}(\x)\}} (\x^{\prime}), \quad i=1,\dots,n.
\end{equation*}
If entries in the weight vector $\xweight$ are constant in $\x$, then, the weight entries in the expression of the respective Metropolis-Hasting ratios $r_{i}(\x,\x^{\prime}), i=1,\dots,n$ cancel, so that each $\xtrans_{i}$ is  $\xtarget$-invariant, and $\xtransc$ is simply a mixture of $\xtarget$-invariant Markov kernels. In contrast, the weights $\xweight$ will not be constant in our examples; and hence,  the  Markov kernels $\xtrans_{i}$ will not generally preserve the target measure $\xtarget$. Instead these kernels can be shown to preserve the probability measures $\xtarget_{i}(\x) \propto \xweight_{i}(\x) \xtarget(\x), \, i=1,\dots,n$, respectively. 
\end{remark}
\subsection{Implementation on the state space graph of a Markov process}
\label{ssec:SkewDetailedBalanceonStateGraph}
While \cref{alg:generalized:MH} is very general, we have not
specified how the involutions $\{S_i\}_{i=1}^{n}$ and the proposal kernels $\{\xprop_i\}_{i=1}^{n}$ may be
chosen, or provided any intuition for why the algorithm might be an improvement over
the classical Metropolis-Hastings algorithm.

In what follows we provide a general
construction which takes a proposal Markov transition kernel $\prop$ and
target measure $\target$ on a discrete
state space $\Domain$ and builds a collection of proposals $\{\xprop_{i}\}_{i=1}^{n}$ and involutions $\{S_i\}_{i=1}^{n}$
on an extended state space $\xDomain$ so that
Algorithm~\ref{alg:generalized:MH} can be used. This construction will
make more precise the idea that the skew Metropolis-Hastings algorithm
adds ``momentum''  to the standard Metropolis-Hastings algorithm. 

Our main conditions on the proposal kernel $\prop$ and $\target$ are summarized in the following assumption.
\begin{assumption}\label{as:prop}
Let the proposal kernel $\prop$ and the target measure $\target$ be such that 
\begin{enumerate}[(i)]
\item \label{it:1:as:prop}$\target(\dm) >0$ for all $\dm \in \Domain$
\item \label{eq:WRev} $\prop(\dm, \dm') \neq 0 \iff \prop(\dm', \dm) \neq 0$
\item\label{it:1:as:prop:ir}
The Markov chain generated by $\prop$ is irreducible.
\end{enumerate}
\end{assumption}
The first condition is mild as states $\dm$ with $\target(\dm)=0$ can
simply be removed from $\Domain$. 
The symmetry condition \ref{eq:WRev} ensures that $\prop$ is equivalent (in the senes that the corresponding transition probabilities have identical support) to some reversible kernel with invariant measure $\target$. Condition \ref{it:1:as:prop:ir} is necessary to ensure that the constructed non-reversible Markov chain is uniquely ergodic (see \cref{sec:ergodic}). 


The proposal kernel induces a graph structure on $\Domain$. States $\dm, \dm^{\prime} \in \Domain$ are said to be adjacent if $\prop(\dm,\dm^{\prime})>0$. We refer to corresponding adjacency graph $\graph = (\vertices, \edges)$ with vertices given as $\vertices=\Domain$ and edges $\edges = \{(\dm,\dm') \in\Domain \times \Domain \mid \,\prop(\dm, \dm')> 0\}$ as the {\em state graph} of the Markov chain; see \cref{sfig:graph:1} for an illustration.   
In the view of this graph structure, the condition \ref{eq:WRev} ensures that the graph $\graph$ is symmetric (or undirected) in the sense that if $(u,v) \in \edges$ then so is $(v,u)\in \edges$, and condition \ref{it:1:as:prop:ir} ensures that $\graph$ is connected.

The general idea of our construction is to build a non-reversible Markov chain by introducing non-reversible flows, typically shaped like ``vortices,'' on the state graph, each of which being associated with a involutive transformation.

Concretely, we begin by specify these flows by a collection of oriented subgraphs $\graph_{i}^{+} = (\vertices_{i},\edges_{i}^{+}), i=1,\dots,n$. We require that each $\graph_{i}^{+}$  has no isolated vertices and that  the symmetric completions $\graph_{i}  = (\vertices_{i},\edges_{i})$, where $\edges_{i} = \edges_{i}^{+} \cup \edges_{i}^{-}$ with $\edges_{i}^{-} = \{ (u,v) \mid (v,u) \in \edges_{i}^{+}\}$, form a cover of $\graph$ in the sense that $\vertices =  \bigcup_{i=1}^{n}\vertices_{i}$,  $\edges =  \bigcup_{i=1}^{n}\edges_{i}$; see \cref{sfig:graph:2}. 

For each oriented subgraphs $\graph_{i}^{+}$, we introduce an auxiliary variable $\vel_{i}$ which takes positive or negative unitary values $\{-1, +1\}$, and we denote the vector of all such auxiliary variables as $\velvec=(\vel_{1},\dots,\vel_{n}) \in \{-1, +1\}^{n}$. In accordance with our notation of \cref{sec:general} we denote the such extended state-space by $\xDomain = \Domain\times\{-1,1\}^{n}$ and we use the shorthand notation $\x = (\dm,\velvec)\in \xDomain$ for elements of that state.

Given this collection of directed subgraphs, our general recipe to built a collections of associated proposals $\{\xprop_{i}\}_{i=1}^{n}$,  involutions $\{S_{i}\}_{i=1}^{n}$ and weights $\{\weight_{i}\}_{i=1}^{n}$ on the extended space $\xDomain$ is as follows.
Let
\begin{equation*}
  \nb_{i}(\dm) := \{\dm^{\prime} \mid (\dm,\dm^{\prime}) \in \edges_{i}\},
\end{equation*}
denote the of neighborhood of $\dm$ in the graph $\graph_{i}$ which are reachable under the proposal kernel $\prop$. 
Let
\begin{equation*}
  \nb_{i}^{+}(\dm) = \{ \dm^{\prime} \in \Domain : ~ (\dm,\dm^{\prime}) \in \edges_{i}^{+} \}, \qquad \nb_{i}^{-}(\dm) = \{ \dm^{\prime} \in \Domain : ~ (\dm,\dm^{\prime}) \in \edges_{i}^{-} \},
\end{equation*}
denote the partition of neighborhood $\nb_{i}(\dm)$
into the set of states which can be reached in one step from the state $\dm$ following the direction of the positive flow $\edges_{i}^{+}$, and the negative flow $\edges_{i}^{-}$, respectively; see \cref{sfig:graph:2}.
We use this partition to built for each subgraph $\graph_{i}$ a proposal kernel $\xprop_{i}$ on $\xDomain_{i}:=\vertices_{i}\times \{-1,1\}^{n} \subset \xDomain$, which for positive value $\vel_{i}=1$ proposes new states in the direction of the positive flow $\edges_{i}^{+}$, and for negative value $\vel_{i}=-1$ proposes new states in the direction of the negative flow $\edges_{i}^{-}$. That is 
\[
\xprop_{i}\big ( (\dm,\velvec), (\ccdot,\velvec) \big ) \propto \mathbbm{1}_{\nb_{i}^{\vel_{i}}(\dm)}(\ccdot) \prop_{i}(\dm,\ccdot),
\]
or, more precisely,
\begin{align} \label{eq:general:constr:1}
\xprop_{i}((\dm, \velvec), (\dm^{\prime}, \velvec^{\prime})) = \begin{cases}
\dfrac{\prop(\dm, \dm^{\prime})}{\prop(\dm,\nb_{i}^{\vel_{i}}(\dm))}  & \text{if  }\dm^{\prime} \in   \nb_{i}^{\vel_{i}}(\dm) \text{ and } \velvec^{\prime} =\velvec, \\
1 & \text{if  } \nb_{i}^{\vel_{i}}(\dm) = \emptyset \text{ and } S_i((\dm^{\prime},\velvec^{\prime})) = (\dm, \velvec),\\
0 & \text{otherwise},
 \end{cases}
\end{align}
where in both the above expressions we used the shorthand notation
\begin{align*}
\nb_{i}^{\vel}(\dm) =\begin{cases}
\nb_{i}^{+}(\dm), &\text{if } \vel= +1\\
\nb_{i}^{-}(\dm), &\text{if } \vel = -1
\end{cases}.
\end{align*}

As a natural choice for the involutive map $S_{i} : \xDomain \rightarrow \xDomain$ we consider the map which flips the sign of the $i$th component of the vector $\velvec$, i.e.,
\begin{align}\label{eq:general:constr:2}
S_i((\dm,\velvec)) = (\dm,R_{i}(\velvec)), \text{ with } R_{i}(\velvec) =  \velvec - 2 \sgn(\velvec\cdot {\rm e}_i) {\rm e}_i,
\end{align}
where ${\rm e}_{i}$ denotes the $i$th canonical vector in $\RR^{n}$. For an illustration, see \cref{fig:stateGraph}.

Throughout the remainder of this article, we assume, that the weight vector $\xweight$ is purely a function of $\dm$, i.e., $\xweight((\dm,\velvec)) = \weight(\dm)$  for some $\weight : \Domain\rightarrow \Delta^{n-1}$ so that \ref{it:cond:weight:2} is trivially satisfied. 
With the $\xprop_{i}$'s and $S_i$'s as defined in \eqref{eq:general:constr:1} and \eqref{eq:general:constr:2}, respesctively, it can be verified that 
\begin{equation}\tag{$\mathcal{C}_{1}'$}\label{eq:cond:weight:1b}
\weight_{i}(\dm) > 0 \iff \dm \in \vertices_{i},
\end{equation}
is sufficient for the remaining conditions \ref{it:cond:weight:1} and \ref{it:cond:weight:2} to be satisfied, provided that the symmetry condition \ref{eq:WRev} of \cref{as:prop} holds. We say that the $i$th proposal is {\em active} in state $\dm$ if $\dm \in \vertices_{i}$, or, equivalently, if $\weight_{i}(\dm)>0$, and we denote by $\activeset(\dm) = \{ i  \in \{1,\dots,n\} \mid \weight_{i}(\dm)>0 \}$ the index set of proposals which are active in $\dm$.

\begin{remark}\label{rem:activeSomewhere}
  Note that by our assumptions on the graphs $\graph_{i}^{+}$, they have non-empty edge sets $\edges_{i}^{+}$. This implies that for every $ i\in\{1,\dots,n\}$ there exists at lease one $\dm \in \Domain$ with $ i \in \activeset(\dm)$. Hence there is always at least one state at which the $i$th momentum can be flipped.
\end{remark}

We consider 
\begin{equation}\label{eq:weights:generic}
\weight_{i}(\dm) = \prop(\dm,\nb_{i}^{+}(\dm)) + \prop(\dm,\nb_{i}^{-}(\dm)),
\end{equation}
as a generic choice for the weight vector which can be easily verified to satisfy condition \eqref{eq:cond:weight:1b}.


Lastly, we extend the definition of the target measure $\target$ from
$\Domain$ to $\xDomain$ as
\begin{equation*}
  \xtarget((\dm,\velvec))= \frac1{2^n}  \target(\dm), \quad \dm \in \Domain,\, \velvec \in \{-1,1\}^{n},
\end{equation*}
so that $\xtarget$ is the product measure of $\target$ and the uniform measure on $\{-1,1\}^{n}$. In particular,
\begin{equation*}
\sum_{\velvec \in \{-1,1\}^n}\xtarget((\dm,\velvec)) =\target(\dm)\,, 
\end{equation*}
i.e., the marginal measure of $\xtarget$ in $\dm$ coincides with the target measure $\target$ on $\Domain$. 
Moreover, with the augmented measure being of product form and uniform in the $\velvec$-component, it follows that the $S_{i}$'s are $\xtarget$-invariant, i.e., $\xtarget(S_{i}((\dm,\velvec)))  = \xtarget((\dm,\velvec) )$  for all $i\in \{1,\dots,n\}$, $\dm \in \Domain$ and
$\velvec,\velvec' \in \{-1,1\}^n$, since  $S_i$ only
changes the $\velvec$-component. 


In conclusion, the collection of proposals, weights and involutions $(\xprop_{i},S_i, \xweight_{i}), i \in \{1,\dots,n\}$ and the augmented measure $\xtarget$ satisfy by construction the  hypotheses of \cref{thm:MSMH}. Hence, with these choices, the MSMH-algorithm (see \cref{alg:generalized:MHb}) produces a Markov chain which preserves the measure $\xtarget$ on $\xDomain$, and thus also the marginal measure $\target$ on $\Domain$.

\begin{figure}
\captionsetup{width=1.0\linewidth}
\subcaptionbox{ Undirected State Graph $\graph = (\vertices,\edges)$ \label{sfig:graph:1}}{\includegraphics[width=.28\linewidth, clip = true, trim = {1.9cm 0cm 2.5cm 0cm}]{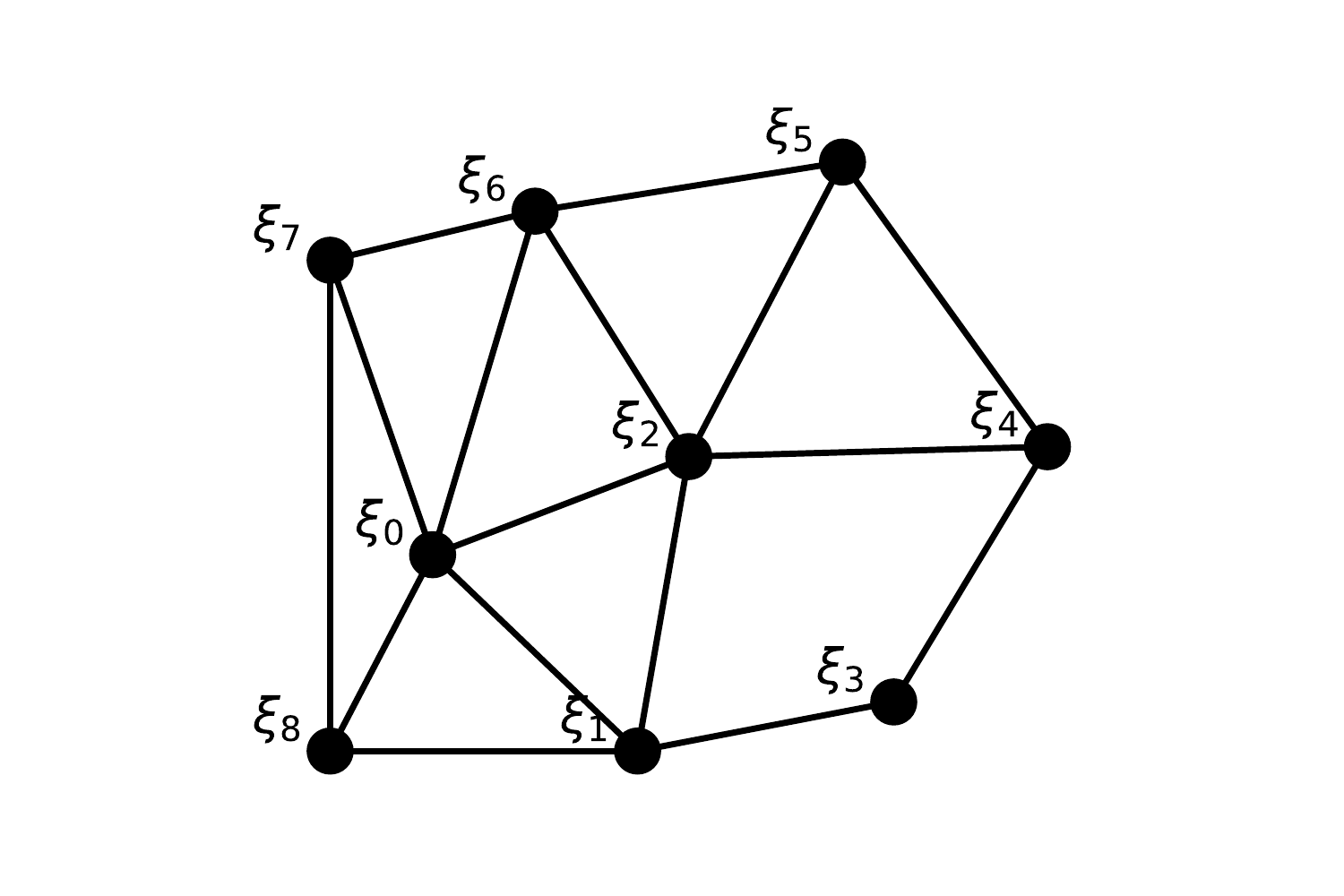}}\qquad
\subcaptionbox{State graph with three assigned flows: $\edges_{1}^{+}$ (blue), $\edges_{2}^{+}$ (red), $\edges_{3}^{+}$ (green) \label{sfig:graph:2}}{\includegraphics[width=0.29\linewidth, clip = true, trim = {1.9cm 0cm 2.25cm 0cm}]{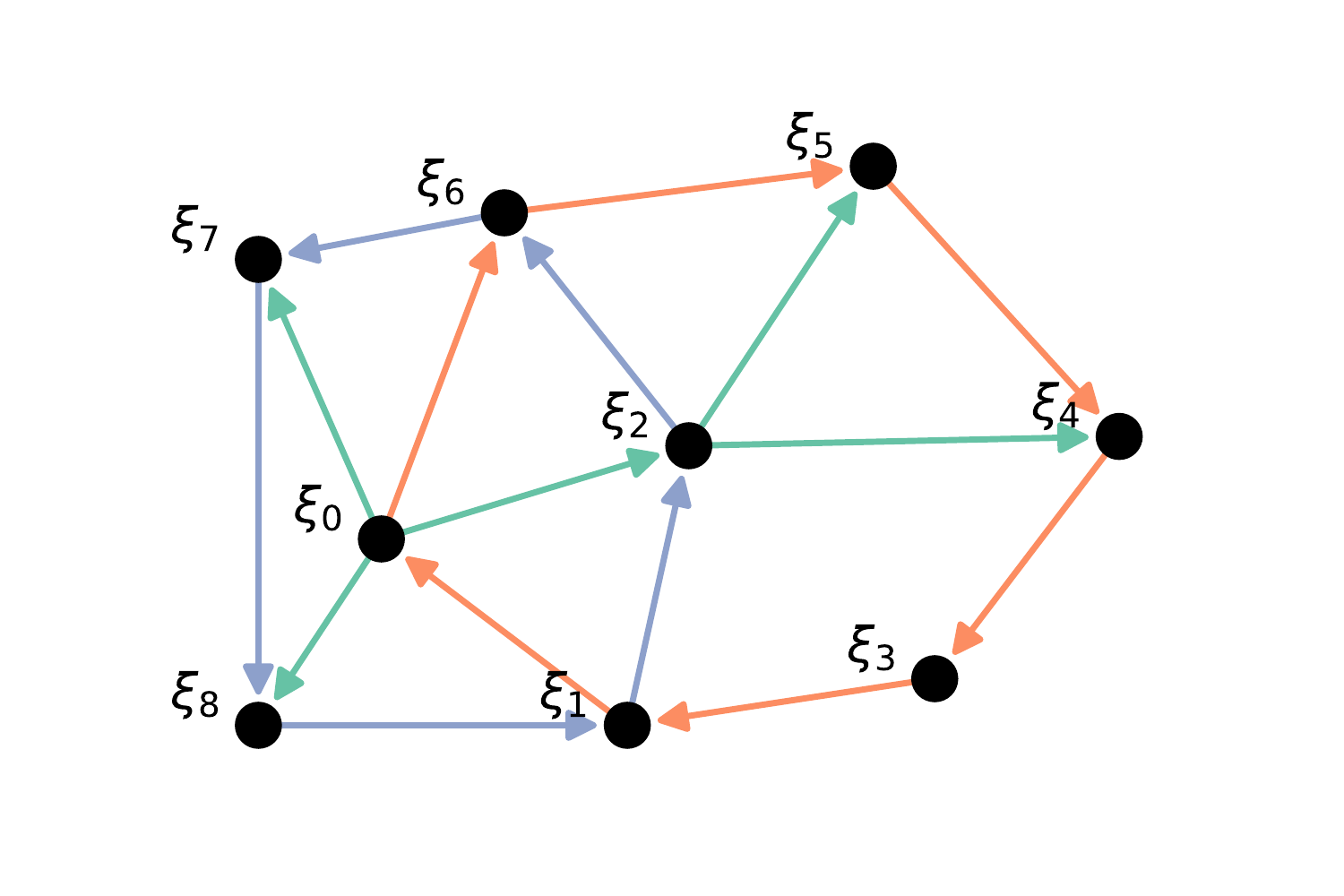}}\qquad
\subcaptionbox{State graph with indicated positive and negative neighborhoods of state $\dm_{0}$. \label{sfig:graph:3}}{\includegraphics[width=0.29\linewidth, clip = true, trim = {1.9cm 0cm 2.25cm 0cm}]{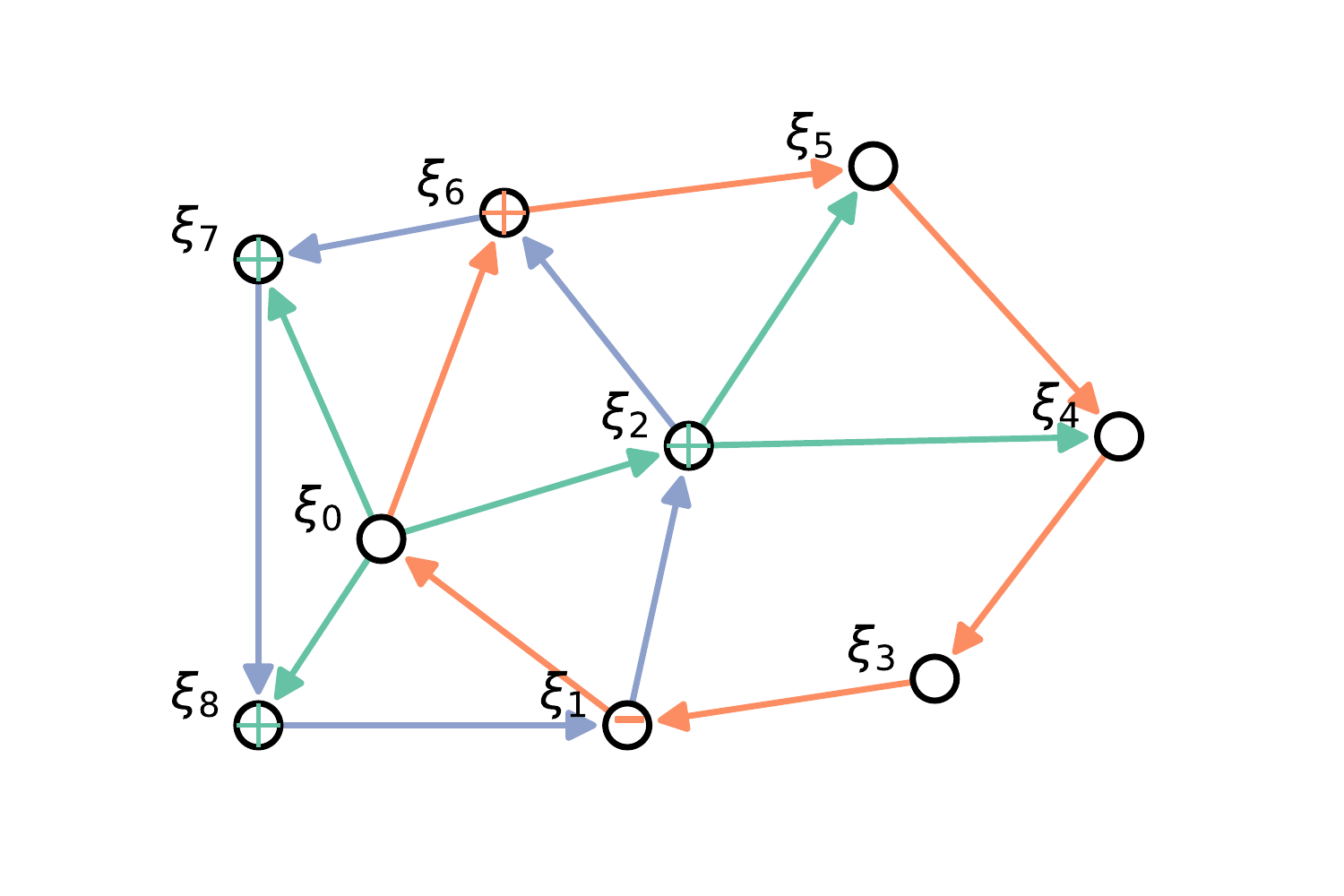}}\qquad
\caption{Panel (A) shows the undirected state graph induced by a reversible proposal kernel $\prop$ on a state space comprised of 9 states. Panel (B) shows an exemplary assignment of flows. 
If the current state is $\dm_{0}$ and either $\velvec=(1,1,1)$ or $\velvec=(-1,1,1)$, we may propose a state from the positive neighborhoods $\nb_{2}^{+}(\dm_{0})=\{ \dm_{6} \} $ or $\nb_{3}^{+}(\dm_{0})=\{\dm_{2},\dm_{7},\dm_{8}\}$. If instead $\velvec=(1,-1,-1)$ or $\velvec=(-1,-1,-1)$, we may propose a state from within the sets $\nb_{2}^{-}(\dm_{0})=\{\dm_{1}\}$; see Panel (C). Note that the neighbourhoods $\nb_{1}^{+}(\dm_{0}), \nb_{1}^{-}(\dm_{0}),$ and $\nb_{3}^{-}(\dm_{0})$ are empty.
 }
\label{fig:stateGraph}
\end{figure}

\begin{remark}[The construction on a General Measure space]
  We have chosen to assume $\Domain$ to be countable. Extending to the 
  case where $\Domain$ is a general separable measure space is straight forward when proposal kernel from state $\dm$, $\prop(\dm,d\dm')$, is absolutely continuous with respect to a common  ($\sigma$-finite) radon measure $\lambda(d\dm')$ for all $\dm \in \Domain$.

  It then follows that the measure $\target$ needs to also be absolutely continuous with respect to $\lambda$. One can then write the detailed balance condition as $\target(d\dm)\prop(\dm,d\dm')=\target(d\dm')\prop(\dm',d\dm)$ as measure on the product space $\Domain\times \Domain$. If we denote by $q(\dm,\dm')$ and $p(\dm)$ to be the densities of $\prop(\dm,d\dm')$ and $\target(d\dm)$ with respect to $\lambda$ then all that follows makes sense and applies to this more general setting with $\prop$ and $\target$ replaced respectively by $q$ and
  $p$. For example the conditions of Assumptions~\ref{as:prop} become $p(\dm) >0$ rather than $\target(\dm)>0$ and $q(\dm,\dm')>0$ rather than  $\prop(\dm, \dm') > 0$. The ratios of $\prop$'s in \eqref{eq:general:constr:1} become ratios of $q$'s evaluated at the same points. The existence of a version (indistinguishable up to null sets) which has the right measurability properties of these ratios and the acceptance rations $r_i$ from Algorithm~\ref{alg:generalized:MH} and~\ref{alg:generalized:MHb} are guarantied by the arguments of Proposition~1 in \cite{Tierney_1998}. 
The fact that graph which is constructed had countable vertices was unimportant as the definitions really do not use the graph structure. All of the definitions make sense once the above modifications have been made.  
\label{rem:genStateSpace}
\end{remark}

\subsection{Ergodic properties}\label{sec:ergodic}
Let in the remainder of this section 
\[
\xtransc\big((\dm,\velvec),\ccdot\big) = \sum_{i=1}^{n} \weight_{i}(\dm) \xtrans_{i}\big((\dm,\velvec),\ccdot\big),
\]
denote the Markov kernel generated by \cref{alg:generalized:MHb} with generic weight function $\weight$ as specified in \eqref{eq:weights:generic}.
If in addition to \cref{as:prop} the following 
\begin{assumption}\label{as:prop:ir}
For any extended state $(\dm,\velvec) \in \dmDomain \times \{-1,1\}^{n}$ and any active index $i \in \activeset(\dm)$, there are  $m\in \NN$ and $\dm^{\prime}\in \dmDomain$ so that
\[
\xtrans_{i}^{m}\big((\dm,\velvec), (\dm^{\prime},R_{i}(\velvec))\big) >0.
\]
\end{assumption}
is satisfied, then we can conclude unique ergodicity of the generated Markov chain as detailed in the following theorem.

\begin{theorem}\label{thm:ergodic:general}
Let \cref{as:prop} and \cref{as:prop:ir} be satisfied. The non-reversible Markov chain $(\x_{k})_{k\in \NN}$ generated by $\xtransc$ and initial extended state $\x_{0}=\tilde{\x}_{0}$ is uniquely ergodic in the sense that
\begin{equation}\label{eq:ergodic:1}
\lim_{N\rightarrow \infty} \frac{1}{N}\sum_{k=0}^{N-1}\mathbbm{1}_{\{\x^{\prime}\}}(\x_{k}) = \xtarget(\x^{\prime})
\end{equation}
 almost surely for any value of the initial state $\tilde{\x}_{0}\in \xDomain$ and all $\x^{\prime}\in \xDomain$.
In particular, for any observable $\varphi : \dmDomain \rightarrow \RR$, we have 
\begin{equation}\label{eq:ergodic:2}
\lim_{N\rightarrow \infty} \frac{1}{N}\sum_{k=0}^{N-1}\varphi(\x_{k}) = \sum_{\x \in \xDomain}\varphi(\x) \xtarget(\x)
\end{equation}
 almost surely for all initial states $\tilde{\x}_{0}\in \xDomain$.
\end{theorem}
\begin{proof}
A proof of this theorem can be found in \cref{sec:proofs}.
\end{proof}

By construction described in \cref{ssec:SkewDetailedBalanceonStateGraph}, \cref{as:prop:ir} implies that for every state $\dm \in \vertices_{i}$ we can reach a vertex $\dm'$ for which the probability of flipping the $i$th velocity component is positive. This can be accomplished by either following a directed path along edges in $\edges_{i}^{+}$ if $\vel_{i}=+1$, or by following a directed path along edges in $\edges_{i}^{-}$ if $\vel_{i}=-1$. Provided that each of the Markov kernels $\{\xtrans_{i}\}_{i=1}^{n}$ possess an invariant measure (which \cref{as:prop} guarantees; see \cref{rem:sqewMix}), we will see bellow that the only obstruction to the existence of  such a reachable vertex $\dm'$ is that $\dm$ is  a vertex of a cycle in the associated directed graph. To make this precise, we need a few concepts.


A {\em circuit} in a directed graph $G=(V,E)$ is a sequence of vertices $(v_{0},\dots,v_{m})$ such that (i) each   pair $(v_{i},v_{i+1}) \in E$ for  $i=0,\dots,m-1$ and (ii) the path begins and ends at the same vertex (i.e. $v_{0}=v_{m}$).
Note that we allow edges $(v_{i},v_{i+1})$  to be repeated. 
We say that a circuit is {\em escapable} if there is at least one edge $(v,v')\in E$ with $v$ being a vertex in the circuit and $v'$ not. Conversely, we say that a circuit  is {\em non-escapable} if such an edge does not exist.  We say that a non-escapable circuit $(v_{0},\dots,v_{m})$ is {\em maximal} if the vertices of the circuit, namely $\{ v_{k} \mid k=0,\dots,m\}$, are not a proper subset of the vertices of another  non-escapable circuit.
%
Equipped with the terminology we can state the following \cref{as:circuit} and \cref{lem:flippable}.



\begin{assumption}\label{as:circuit}
For every value of $\velvec \in \{-1,1\}^{n}$ and $i \in \{1,\dots,n\}$ every maximal non-escapable circuit 
in $\graph^{\vel_{i}}_{i}$ contains at least one state $\dm'$ for which $P_{i}\big(\,(\dm',\velvec), (\dm',R_{i}(\velvec))\,\big)>0$. Here we use the short-hand notation 
\[
\graph^{\vel_{i}}_{i}  = \begin{cases}
\graph^{+}_{i}&, \vel_{i} = +1\\
\graph^{-}_{i}&, \vel_{i} = -1\\
\end{cases}.
\]
\end{assumption}

\begin{theorem}\label{lem:flippable}
Provided that \cref{as:prop} holds, then \cref{as:circuit}  and  \cref{as:prop:ir} are  equivalent. In short: (\cref{as:prop} + \cref{as:circuit}) $\iff$ (\cref{as:prop} + \cref{as:prop:ir}).
\end{theorem}
\begin{proof}
We prove this theorem in \cref{sec:proofs}.
\end{proof}


In practice, \cref{as:prop:ir} and even \cref{as:circuit} may often be difficult to check. In order to guarantee ergodicity in the sense of \cref{thm:ergodic:general}, one may instead consider a ``lazy'' version of the algorithm, where in each step 
\begin{enumerate}
\item \label{it:mod:lazy:1} with some small probability $\varepsilon>0$, a component index is uniformly sampled from $\{1,\dots,n\}$ and the sign of the corresponding velocity component flipped, while the state $\dm$ remains unchanged.  
\item  \label{it:mod:lazy:2} with probability $1-\varepsilon$, the steps in \cref{alg:generalized:MHb} are executed.    
\end{enumerate}
Since the Markov kernel of the Markov chain generated by this modification of the algorithm is a mixture of two kernels which each preserve $\xtarget$ it directly follows that this Markov chain indeed preserves $\xtarget$. Moreover, provided that $\prop$ is irreducible it follows by similar arguments as in the proof of \cref{thm:ergodic:general} that the such generated non-reversible Markov chain is also irreducible, which is sufficient for the Markov chain to be uniquely ergodic in the sense specified in \cref{thm:ergodic:general}.
In addition we may consider an even ``lazier'' modification of the algorithm where in each step in addition to the above described modification, we may leave with non-zero probability  the complete extended state $(\dm,\velvec)$ unchanged. It is easy to see that under this modification, the resulting Markov chain will not only preserve the target measure and be irreducible, but will in addition also be guaranteed to be acyclic, so that the chain converges in law to the target measure, i.e., $\lim_{N \rightarrow \infty} \xtransc^{N}\big((\dm,\velvec),(\dm',\velvec')\big) = \xtarget(\dm',\velvec')$ for all $(\dm,\velvec)\in \dmDomain \times \{-1,1\}^{n}$ and all $(\dm',\velvec') \in \dmDomain \times \{-1,1\}^{n}$.\\




In what follows we show how the framework described in this section can be applied to devise non-reversible MCMC methods for the sampling of redistricting plans. For this purpose we describe in the following \cref{sec:space:and:distribution} the sampling space $\Domain$  of that application --the set of redistricting plans-- and the target measure $\target$ defined on that space.  We then introduce in \cref{sec:SingleNodeFlip} a reversible Markov kernel $\prop$ on set of redistricting plans, which we construct as a tempered version of the single node flip algorithm of \cite{herschlag2018quantifying,jcmReport}. In \cref{sec:AppToRedistFormal} we describe and discuss three methods to partition and direct the induced state space graph $\graph = (\vertices,\edges)$.

\section{Sampling redistricting plans}\label{sec:application}

\subsection{The space and distribution of redistricting plans}\label{sec:space:and:distribution}
Throughout this article we consider the problem of sampling redistricting plans as equivalent to sampling the space of $\ndist$ partitions on a graph $\pgraph = (\pvertices, \pedges)$, which is a discrete representation of the administrative region (e.g., a state, or collection of counties) for which redistricting plans are drawn. The vertices ($\pvertices$) of the graph correspond to geographic regions of a certain administrative level (e.g., voter tabulation districts (VTD), precincts, census blocks), and edges ($\pedges$) are placed between vertices that are either rook, queen, or legally adjacent\footnote{Rook adjacency means that the geographical boundary between two regions has non-zero length; queen adjacency means that the boundaries touch, but may do so at a point. At times two regions may not be geographically adjacent, but may be considered adjacent for legal purposes; for example, an island may still be considered adjacent to regions on a mainland for the purposes of making districts.}; see \cref{subfig:precinct} (depending on the situation). 

In order to keep language simple we present our approach in the setup where the vertex set $\pvertices$ represent the {\em precincts} of a {\em state}, and we refer to $\pgraph$ as the {\em precinct graph}.

We represent a single districting plan made up of $\ndist$ districts as a function $\dm: \pvertices \to \{1,2 \dots \ndist\}$. In other words, districting plans are $\ndist$-coloring of the graph $\pgraph$. Informally, $\dm(v)=i$ means that the precinct associated with vertex $v$ is in the $i$th district. Given a districting plan $\dm$, we denote by 
\[
D_{i}(\dm)= \{v \in \pvertices \mid \dm(v) = i\},\quad E_i(\dm) = \{(u,v) \in \pedges \mid \dm(v) = \dm(u) = i\},
\]
the set of precincts assigned to the $i$th district, and the set of edges between vertices corresponding to precincts within the $i$th district, respectively. 

On the set of all districting plans $\big \{\dm \mid \dm : \pvertices \to \{1,2 \dots \ndist\} \big \} $, we define a score function $J$ which measures how well a redistricting plan complies with a set of prescribed criteria such as compactness of districts, equal partition of the population or preservation of municipalities. 


This score function is used to define a probability measure $\target$ on the set of districting plans via the relation
\begin{align}
\label{eq:gibbs:measure}
\target(\dm) \propto e^{-J(\dm)}.
\end{align}
The lower the score $J(\dm)$ of a redistricting plan, the better it complies with the criteria and the higher is the probability assigned to it by the Gibbs measure $\target(\dm)$. In particular, if a redistricting plan $\dm$ is non-compliant, then, $J(\dm)=\infty$, and thus $\target(\dm)=0$. 

In the remainder of this article we denote by
\begin{equation*}
\dmDomain := \{ \dm : \pvertices \rightarrow \{1,\dots,\ndist\} \mid J(\dm) < \infty \}.
\end{equation*}
the support of $\target$, to which in the following we will also refer as the set of all {\em possible maps} or {\em possible redistricting plans}.

\begin{figure}
\captionsetup{width=1.0\linewidth}
\subcaptionbox{Precinct graph\label{subfig:precinct}}{\includegraphics[width=\dmwidth \linewidth, clip = true, trim = {0cm 3cm 0cm 3cm}]{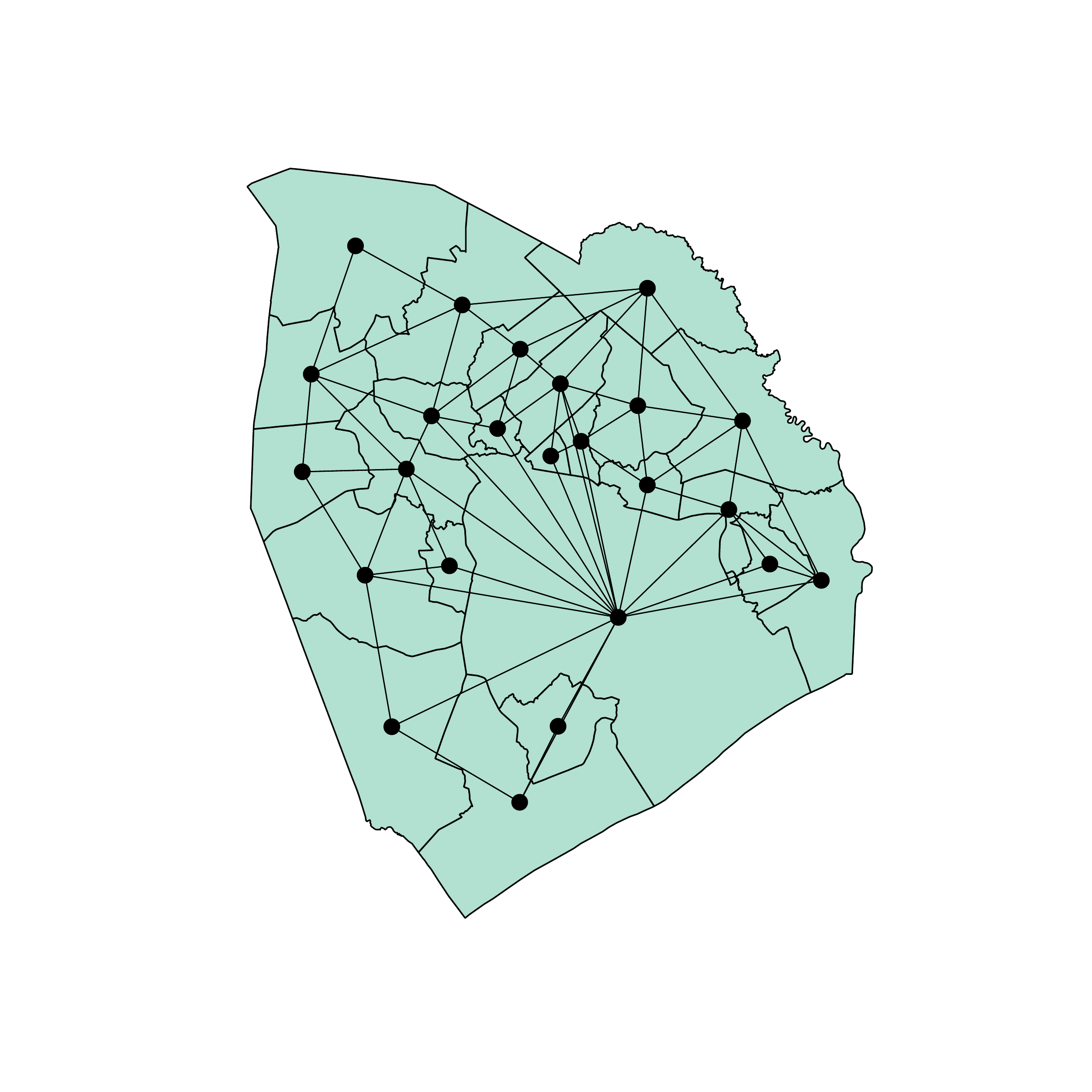}}\qquad
\subcaptionbox{District graph\label{subfig:dg:1}}{\includegraphics[width=\dmwidth\linewidth, clip = true, trim = {0cm 3cm 0cm 3cm}]{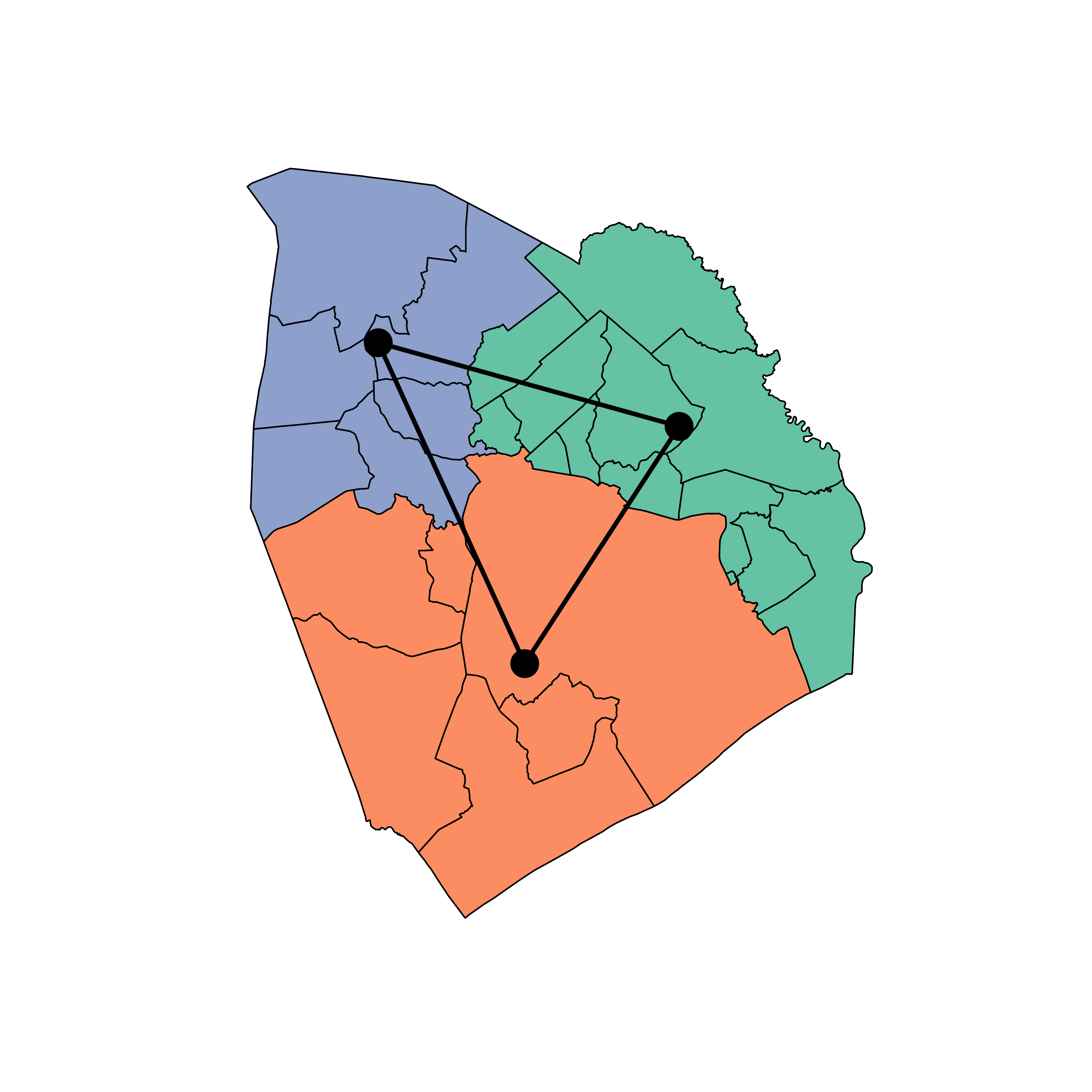}}
\caption{(A) State  
with associated precinct graph; the embedding $\embedding$ used for displaying the precinct graph places nodes at the areal centroid of the associated precinct. (B) Coloring of the precinct graph which corresponds to a partition of the state into three districts. Superimposed is the associated district graph.
 }
\end{figure}

\subsection{The score function}
The score function $J(\dm)$
which determines the measure $\target$ typically relies on additional information associated with vertices and edges of the precinct graph 
$\pgraph$ such as population, land area, and border length. This additional data is used to evaluate the districts on desired redistricting criteria, such as equal-population and compactness.

We denote by $\pop(v)$ and $\area(v)$ the population and area, respectively,  of the geographical region corresponding to vertex $v \in \pvertices$. Similarly, for $e = (u,v) \in \pedges$, we denote by $\boundary(e)$ the length of the boundary shared between the precincts $u$ and $v$. Certain vertices may not share all of their boundary with an adjacent vertex; for example, a vertex may be on the boundary of the map.  In this case, we also describe the unshared boundary of a vertex to be $\boundary(v)$ (which will be zero for interior nodes).

On the set of all possible maps, the score function $J$ may be constructed as a positive linear combination of sub-functions, each of which evaluate different properties of the redistricting map (see, e.g., \cite{herschlag2018quantifying}). For example, in some of the numerical examples of \cref{sec:num}, we let 
\begin{equation}\label{eq:score:example}
J(\dm) =  w_{pop}J_{pop}(\dm) + w_{c}J_{c}(\dm),
\end{equation}
where $J_{pop}(\dm)$ is a measure of the population deviation and $J_{c}(\dm)$ is a measure of how compact the districts are. Typically $J_{pop}(\dm)$ is in the form of a hard constraint or a sum of squared deviations of each district from a target population; similarly, $J_{c}(\dm)$ is typically a sum of district isoparametric ratios or some measure of the overall perimeter.

\subsection{The tempered Single Node Flip proposal and algorithm}
\label{sec:SingleNodeFlip}
To utilize the ideas of Section~\ref{ssec:SkewDetailedBalanceonStateGraph} in order to sample from $\target$ on the space of possible redistricting plans, we must first establish a  proposal kernel $\prop$ on  the domain of possible redistricting plans $\dmDomain$ which satisfies \cref{as:prop}. 

For redistricting problems, one of the most widely used methods is what is commonly known as the {\em single node flip algorithm}. This algorithm has been shown to mix well on smaller problems \cite{fifield2015new,jcmRebuttal, jcmReport}, but as the size of the districting plan and the criteria for redistricting becomes more complex, the moving boundary MCMC algorithms will converge, in theory, but the mixing time for these chains may cause their use to be infeasible to solve computationally \cite{njatDedfordSolomon2019graphs}.

In this article we use the proposal kernel $\prop$  of a  tempered version  of the single node flip algorithm as the basis for constructing our non-reversible MCMC methods. As in the classical version of the algorithm, the proposal kernel of this variant of the algorithm changes ``flips'' the color of exactly one vertex on the boundary of a district to the color of a neighboring district.


More specifically, for an ordered pair, $(u,v) \in \pvertices \times \pvertices$, of two distinct precincts we define the {\em flip operator}
\begin{equation*}
F_{(u,v)} : \dmDomain \rightarrow \dmDomain, ~~F_{(u,v)}(\dm)(w) = \begin{cases}
\dm(w), & w \neq v,\\
\dm(u), & w = v,
\end{cases}
\end{equation*}
which flips the label of vertex $v$ to the label of vertex $u$, and where as above $\dmDomain$ is the domain of all possible maps. 
For a given districting map $\dm\in \dmDomain$, we denote by 
\[
\cedges(\dm) :=\big\{ (u,v) \in \pedges : \dm(u) \neq \dm(v), ~F_{(u,v)}(\dm) \in \dmDomain  \big\},
\]
 the set of all {\em conflicting edges}, i.e., edges which connect precincts with different labels (precincts which are assigned to different districts), and for which application of the corresponding flip operator results in a valid redistricting plan. Moreover, denote by
\[
\nb(\dm) := \big\{F_{(u,v)}(\dm)  |  (u,v)\in \cedges(\dm) \big\},
\]
the set of all possible redistricting plans which can be obtained from the districting plan $\dm$ upon application of the flip operator along a conflicting edge; that is the vertices of the neighborhood of $\dm$. 
With this notation at hand we define the proposal distribution $\prop(\dm,\ccdot)$ to be a tempered version of the target measure $\target$ constrained to the set $\nb(\dm)$, i.e.,
\[
\prop(\dm,\ccdot) \propto \mathbbm{1}_{\nb(\dm)}(\ccdot) {\rm e}^{-\beta J(\ccdot)}, \quad \dm \in \dmDomain,
\]
or, more explicitly, 
\begin{equation}\label{eq:prop:sn}
\prop(\dm,\dm^{\prime})= \begin{cases}
\frac{1}{Z_{\beta}(\dm)} {\rm e}^{-\beta J(\dm^{\prime})}, & \dm \in \nb(\dm)\\
0, & \text{otherwise}
\end{cases}
\end{equation}
where $\beta>0$ can be understood as an inverse temperature parameter, and 
\[
Z(\dm) := \sum_{\dm^{\prime} \in \nb(\dm)} {\rm e}^{-\beta J(\dm^{\prime})}
\]
is the corresponding partition function.



The tempered proposal kernel $\prop$ can be combined with a Metropolis acceptance rejection criteria to obtain a reversible Metropolis-Hastings algorithm (see \cref{alg:single:node:flip:confl}) which possess $\target$ as an invariant measure. 

\begin{figure}
\captionsetup{width=1.0\linewidth}
{\includegraphics[width=\dmwidth\linewidth, clip = true, trim = {0cm 3cm 0cm 3cm}]{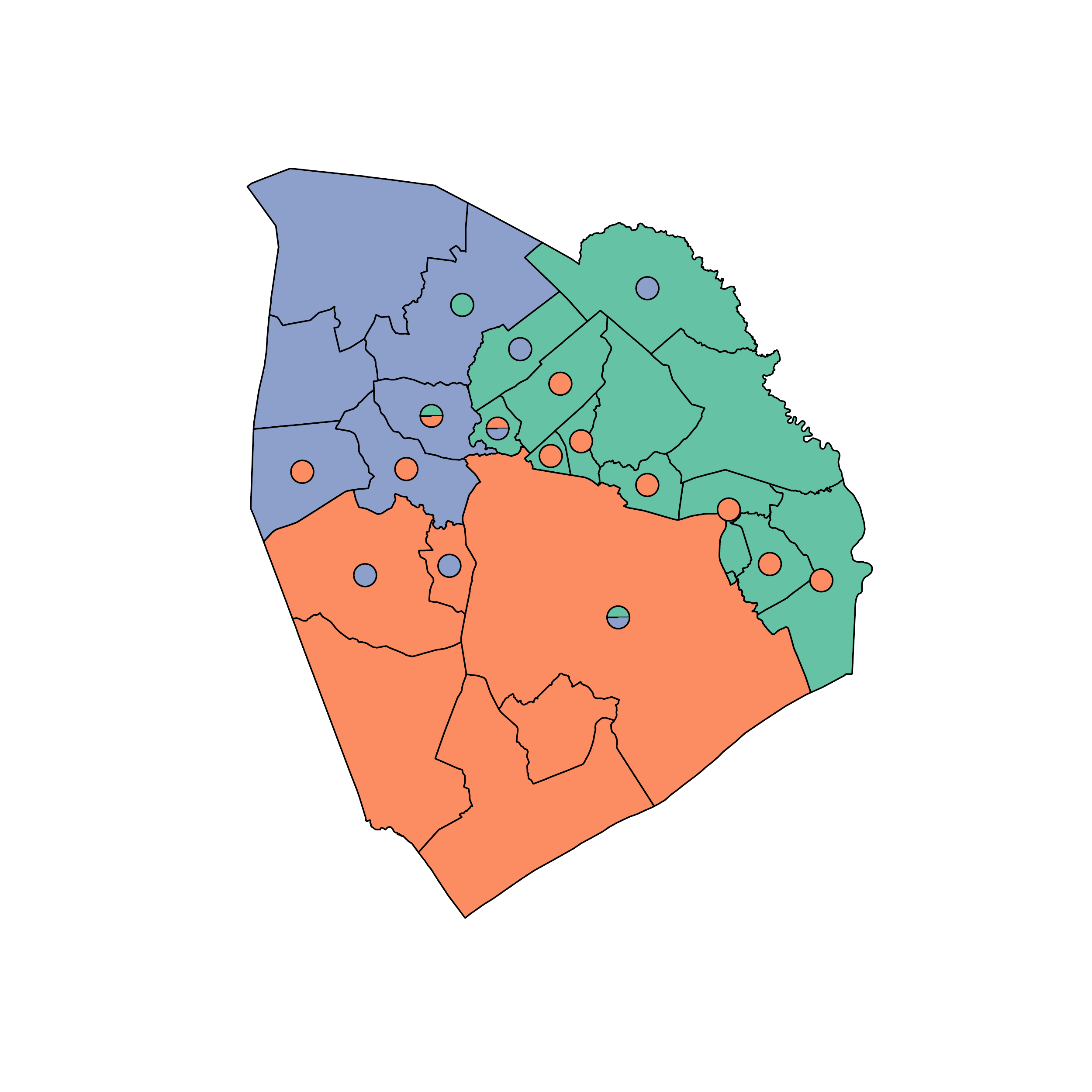}}
\caption{Support of the single node flip proposal; the color(s) of the circular markers within precincts indicate that the probability of proposing the districting plan which is obtained from the current districting plan by changing the color of the respective precinct to (one of) the color(s) of the marker, is positive.  \label{fig:sn}}
\end{figure}

\section{Application to graph partitions and redistricting}
\label{sec:AppToRedistFormal}
In this section we present two different approaches for applying the framework of \cref{ssec:SkewDetailedBalanceonStateGraph} to built non-reversible Markov chains for the sampling of redistricting maps out of the tempered single node flip algorithm described in \cref{sec:SingleNodeFlip}. 
We describe these two approaches in \cref{sec:precinct:flow} and \cref{sec:Dist2Dist}, respectively. In \cref{subsec:comp:temp} we discuss computational aspects of the algorithms and the role of tempering/choice of the temperature parameter $\beta$.
\subsection{Precinct flows via planar embeddings}\label{sec:precinct:flow}
In this approach we construct a single flow via a planar embedding of the planar graph in such a way that state transitions in positive flow direction are aligned  are aligned with a prescribed vector field in $\RR^{2}$. 


Consider the embedding  $\embedding : \pvertices \rightarrow \RR^{2}$  of the precinct graph in $\RR^{2}$, and a volume preserving vector field $\mathbf{v}$ in $\RR^{2}$. 

More specifically, we consider the embedding which maps every vertex $v$ to the center of mass (assuming constant density) of the associated precinct and  the field defined by concentric circles oriented in the counter-clockwise direction. 

That is, 
\[
\embedding(v) = \frac{1}{{\rm area}(v)}\int_{A_{v}} x\, \dd x 
\]
where $A_{v} \subset \RR^{2}$ is the precinct represented by $v$ in some suitable geographical map representation of the state, and
\[
\mathbf{v}(x,y)  = (-r\sin(\alpha), r\cos(\alpha)),
\]
 where $(r,\alpha) = (\sqrt{x^{2}+y^{2}}, \arctan(x/y))$ are the polar representations of the $(x,y)\in \RR^{2}$.
 
For a given embedding $\embedding$ and vector field $\vf$, we need to make precise what is meant by a transition to be aligned with the vector field $\mathbf{v}$. For this purpose we need to define a function, $\orientation : \edges \rightarrow \{-1,1\}$, which maps every edge of the state graph $\edges$ of the Markov chain to the set $\{-1,+1\}$, with $+1,-1$ indicating an alignment with the vector field in positive and negative direction, respectively. This then induces a flow on the state graph as
\[
\edges^{+} = \{ (\dm,\dm^{\prime}) \mid \orientation(\dm,\dm^{\prime})>0\}.
\] 

The question now is how to construct such a function $\orientation$. 
To accomplish this we base positively align the movement of centroids with the vector field.
Let
\begin{align}
\centroid(D_{i}(\dm)) = \frac{1}{\area(D_{i}(\dm))}\sum_{v \in D_{i}(\dm)} \area(v) \embedding(v),
\end{align}
denote the centroid of the $i$-th district. We orient edges $(\dm,\dm^{\prime})$ such that transitions from $\dm$ to $\dm^{\prime}$ are such that the movements of the centroids of the involved districts $i,j$ (these are the districts for which either a precinct is removed or added in the course of the transition) are aligned with the direction of the vector field ${\mathbf{v}}$, i.e.,
\begin{equation}\label{eq:def:orientation:2}
\orientation(\dm,\dm^{\prime}) =  {\rm sign} \left ( \sum_{k\in\{i,j\}} \vectorfield \left ( \frac{ \centroid(D_{k}(\dm^{\prime})) +  \centroid(D_{k}(\dm))}{2} \right )  \cdot \left [ \centroid(D_{k}(\dm^{\prime})) -  \centroid(D_{k}(\dm)) \right ]\right ),
\end{equation}
where $i,j$ are the indices of the two districts which are modified in the transition from $\dm$ to $\dm^{\prime}$. See \cref{alg:vf} for an algorithmic implementation and \cref{fig:cmf} for a graphical illustration of this method.

\begin{figure}
\begin{minipage}[t]{0.68\textwidth}
\begin{algorithm}[H]
\caption{Center-of-mass flow}\label{alg:vf}
\SetKwInOut{Input}{input}
\SetKwInOut{Output}{output}
\Input{$\dm,\vel$}
$\nb^{\vel}(\dm) \gets \{ \dm^{\prime} \in \nb(\dm) : \orientation(\dm,\dm^{\prime})=\vel \}$\;
\eIf{$\nb^{\vel}(\dm) = \emptyset$}{
$\vel \gets -\vel$
}
{
sample $\dm^{\prime} \sim  \mathbbm{1}_{\nb(\dm)}(\ccdot) {\rm e}^{-\beta J(\ccdot)}$,\quad sample $u \sim \mathcal{U}([0,1])$\;
$\nb^{-\vel}(\dm^{\prime}) \gets \Big\{ \dm^{\prime\prime} \in \nb(\dm^{\prime}) : \newline \hspace*{7em}\orientation(\dm^{\prime},\dm^{\prime\prime})=~-\vel\Big\}$\;
\eIf{$\nb^{-\vel}(\dm^{\prime})=\emptyset$ {\bf or} $u <
\dfrac{ 
e^{J(\dm )-J(\dm^{\prime})} Z^{\vel}(\dm)
}
{ 
 e^{\beta J(\dm) - \beta J(\dm^{\prime})}Z^{-\vel}(\dm^{\prime})
}
$}{
$\dm \gets \dm^{\prime}$
}{
$\vel \gets -\vel$
}
}
\KwRet{$\dm,\vel$}
\end{algorithm}
\end{minipage}
\caption*{
In  the edge-aligned version of the algorithm, $\orientation(\dm,\dm^{\prime})$ has the form specified in \eqref{eq:def:orientation:2}.
}
\end{figure}

\vspace{.5cm}
\begin{figure}
\captionsetup{width=1.0\linewidth}
\subcaptionbox{ Center-of-mass-flow, conceptual sketch \label{subfig:cmf:1}}{\includegraphics[width=\dmwidth\linewidth, clip = true, trim = {0cm 3cm 0cm 3cm}]{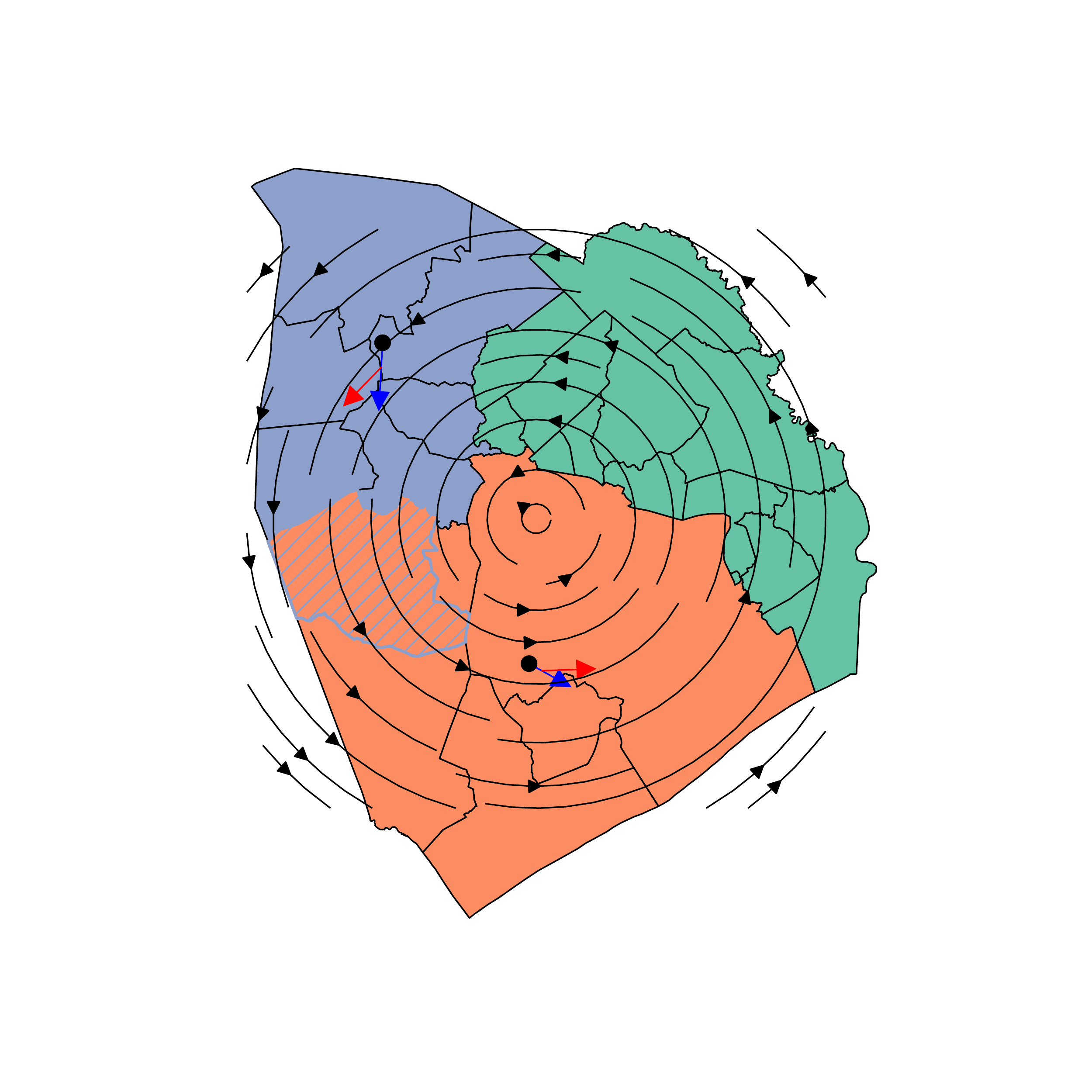}}
\subcaptionbox{Support of center-of mass-flow proposal \label{subfig:cmf:2}}{\includegraphics[width=\dmwidth\linewidth, clip = true, trim = {0cm 3cm 0cm 3cm}]{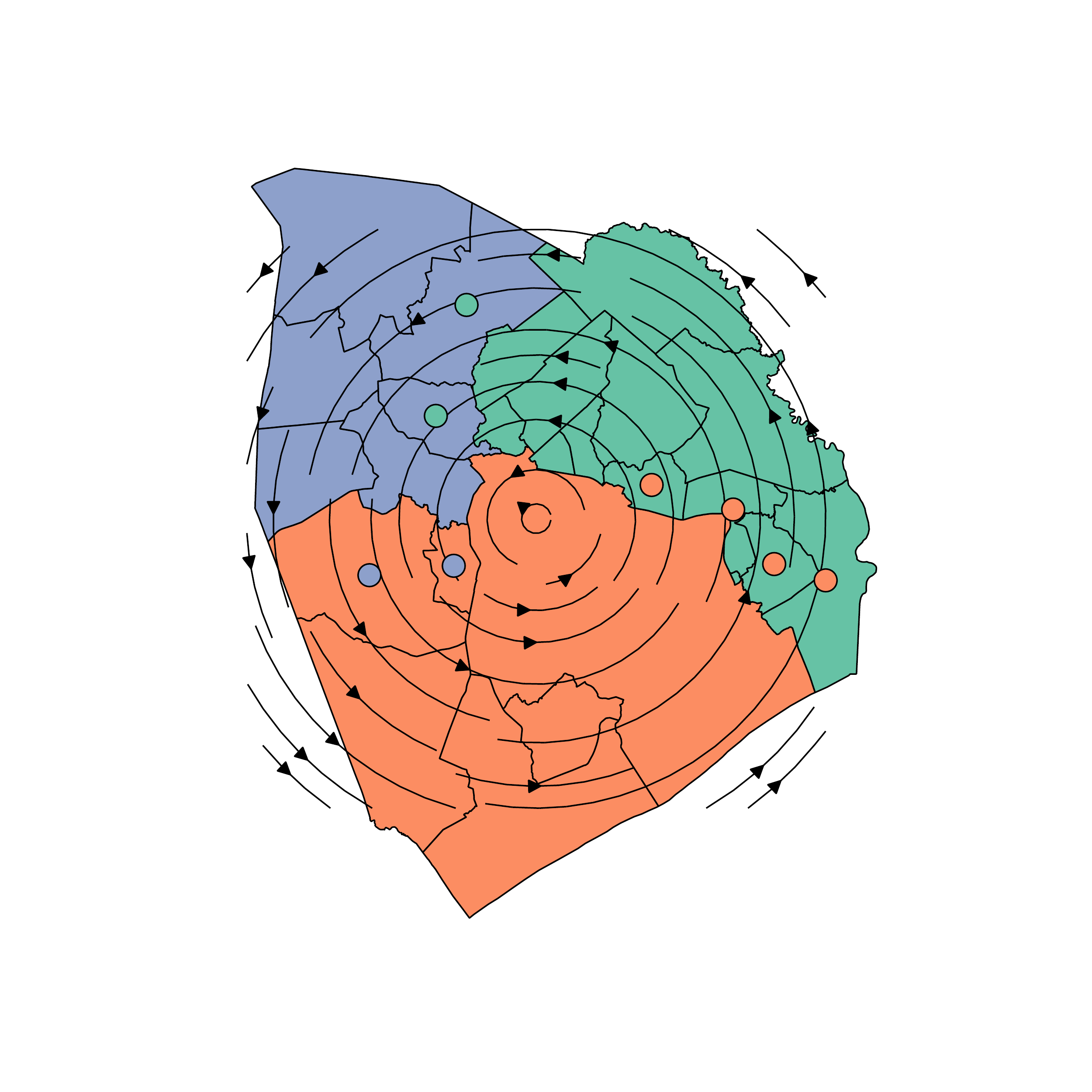}}
\caption{
(A) Geometrical illustration of the orientation assigning function $\orientation$ in the center-of-mass flow method: in a transition where the color of the hatched precinct changes from red to blue the centroids of the corresponding districts (black dots) change as indicated by the blue arrows. The red arrows correspond to evaluations of the vector field (black arrows) at the respective midpoints of the blue arrows. The orientation of the transition is computed as the sign of the sum of the inner products of the two arrow/vector pairs of each district.  (B) Support of the center-of-mass-flow proposal distribution for positive velocity $\vel=1$; The same convention is used for displaying the support as in \cref{fig:sn}.   
}
\label{fig:cmf}
\end{figure}

\subsection{District to district flows}
\label{sec:Dist2Dist}
We associate $n=\ndist(\ndist-1)$ momenta $\vel_{e} \in \{-1,1\}, e \in\dedges$ across the ordered district pairs $\dedges = \{(i,j) \in \{1,\dots,\ndist\}^{2} : i <j  \}$. If the momentum $\vel_{e}$ associated with an adjacent district pair $e=(i,j)$  is positive, we may only propose state changes in which a boundary precinct of district $j$ is reassigned to district $i$. If, on the other hand, $\vel_{e}=-1$,  we may only propose state changes in which a boundary precinct of district $i$ is reassigned to district $j$; see \cref{subfig:dgf:1}.


In the view of \cref{ssec:SkewDetailedBalanceonStateGraph} the construction of the algorithm is as follows. The velocity vector  is of the form $\velvec = (\vel_{e_{1}},\dots,\vel_{e_{n}}) \in \{-1,1\}^{\dedges}$. For any pair of districts $e=(i,j)\in \dedges$, the corresponding positive flow $\edges_{e}^{+}$ of  transitions between possible districts where a precinct from district $i$ is removed and added to district $j$  is 
\[
\edges_{e}^{+} = \{ (\dm,F_{(u,v)}(\dm))  \mid \dm \in \edges, ~  (u,v) \in C_{e}^{+}(\dm)  \}, 
\]
where 
\[
C_{e}^{+}(\dm) = \{ (u,v) \in \conflict(\dm) \mid u \in  D_{i}(\dm), v \in  D_{j}(\dm) \}, 
\]
denotes the set of directed conflicted edges which connect a precinct of district $i$ with a precinct of district $j$. 
Similarly, we have $\edges_{e}^{-} = \{ (\dm,F_{(u,v)}(\dm))  \mid \dm \in \edges, ~  (u,v) \in C_{e}^{-}(\dm)  \}$ where   $C_{e}^{-}(\dm) = \{(u,v) \mid (v,u) \in C_{e}^{+}(\dm) \}$. The $e$th vicinity of the state $\dm$ in positive and negative direction can be explicitly written as
\[
\nb_{e}^{+}(\dm) = \{ F_{(u,v)}(\dm) \mid ~ (u,v)  \in C_{e}^{+}(\dm) \}, \quad \nb_{e}^{-}(\dm) = \{ F_{(u,v)}(\dm) \mid ~ (v,u)  \in C_{e}^{+}(\dm) \},
\]
respectively, and we can write the $e$th proposal kernel on the extended space $\dmDomain \times \{-1,1\}^{\dedges}$ as 
\begin{align*}
\xprop_{e}((\dm, \velvec), (\dm^{\prime}, \velvec^{\prime})) = \begin{cases}
\dfrac{e^{-\beta J(\dm^{\prime})}}{Z^{\vel_{e}}_{e}(\dm)}  & \text{if  }\dm^{\prime} \in   \nb_{e}^{\vel_{e}}(\dm) \text{ and } \velvec^{\prime} =\velvec, \\
1 & \text{if  } \nb_{e}^{\vel_{e}}(\dm) = \emptyset \text{ and } \dm^{\prime} = \dm,\, \velvec^{\prime} = F_{e}(\velvec),\\
0 & \text{otherwise},
 \end{cases}
\end{align*}
where $Z^{\vel_{e}}_{e}(\dm) =  \sum_{\dm^{\prime} \in \nb_{e}^{\vel_{e}}(\dm)} {\rm e}^{-\beta J(\dm^{\prime})}$. 
The generic choice \eqref{eq:weights:generic} for the weight vector $\weight(\dm) = (\weight_{e_{1}}(\dm),\dots,\weight_{e_{n}}(\dm)) \in \{-1,1\}^{\dedges}$ results in weights of the form
\begin{equation}\label{eq:weights:distr}
\weight_{e}(\dm) = \frac{Z_{e}(\dm)}{Z(\dm)}, \quad\text{where}\quad Z_{e}(\dm) = Z^{\vel_{e}}_{e}(\dm) + Z^{-\vel_{e}}_{e}(\dm),
\end{equation}
for $e \in \dedges$, so that $Z(\dm) = \sum_{\tilde{e} \in \dedges} Z_{\tilde{e}}(\dm)$. With this choice of proposal kernel and weights, the Metropolis-Hastings ratio becomes
\[
r_{e}((\dm,\velvec),(\dm^{\prime},\velvec^{\prime})) = \frac{ 
e^{J(\dm )-J(\dm^{\prime})} Z_{e}^{\vel_{e}}(\dm)
}
{ 
 e^{\beta J(\dm) - \beta J(\dm^{\prime})}Z_{e}^{-\vel_{e}}(\dm^{\prime})
}.
\]
We provide an explicit implementation as \cref{alg:single:node:flip:distr}.

\begin{figure}
\captionsetup{width=1.0\linewidth}
\subcaptionbox{District-to-district flow\label{subfig:dgf:1}}{\includegraphics[width=\dmwidth\linewidth, clip = true, trim = {0cm 3cm 0cm 3cm}]{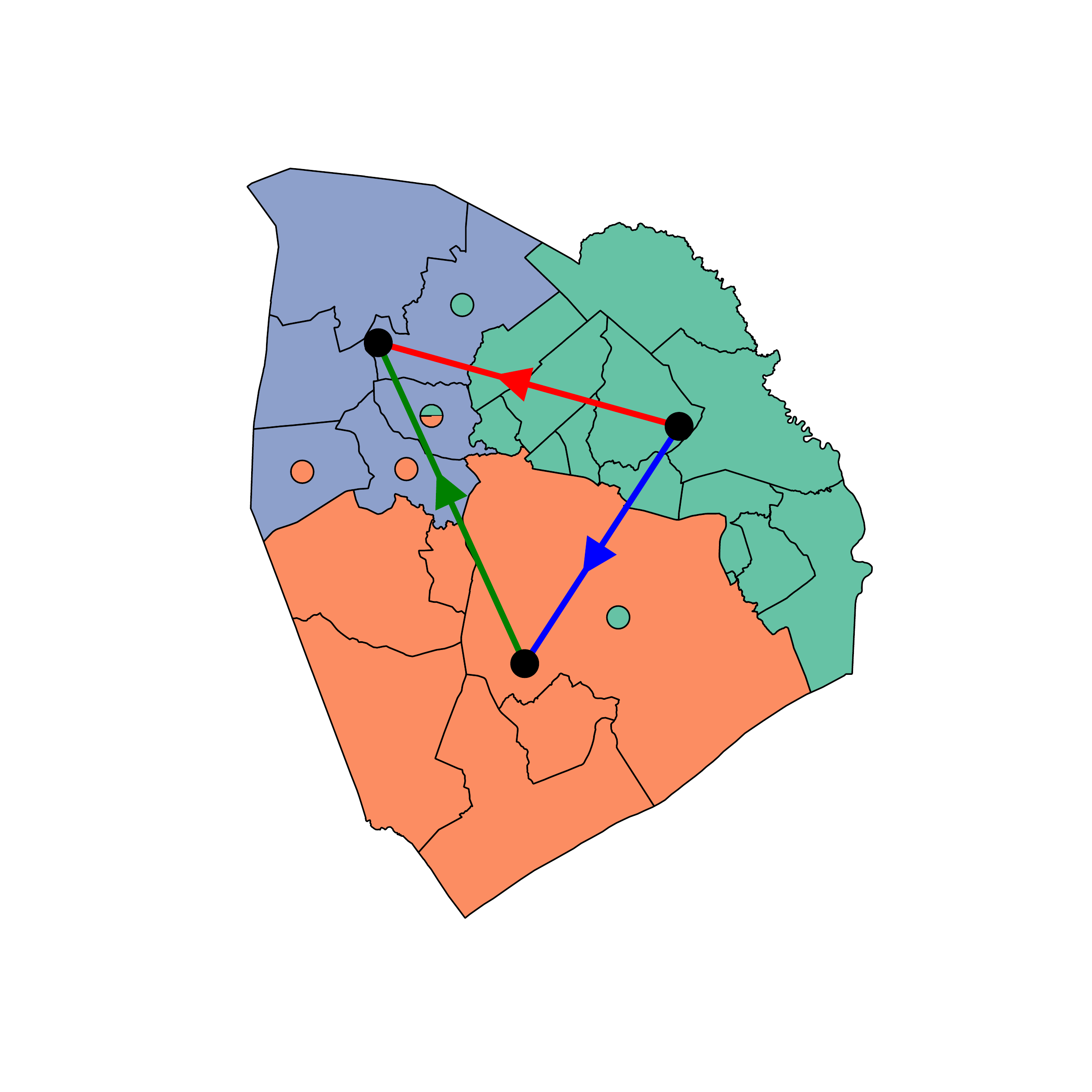}}\qquad
\subcaptionbox{Circular flow on district graph\label{subfig:dgf:2}}{\includegraphics[width=\dmwidth\linewidth, clip = true, trim = {0cm 3cm 0cm 3cm}]{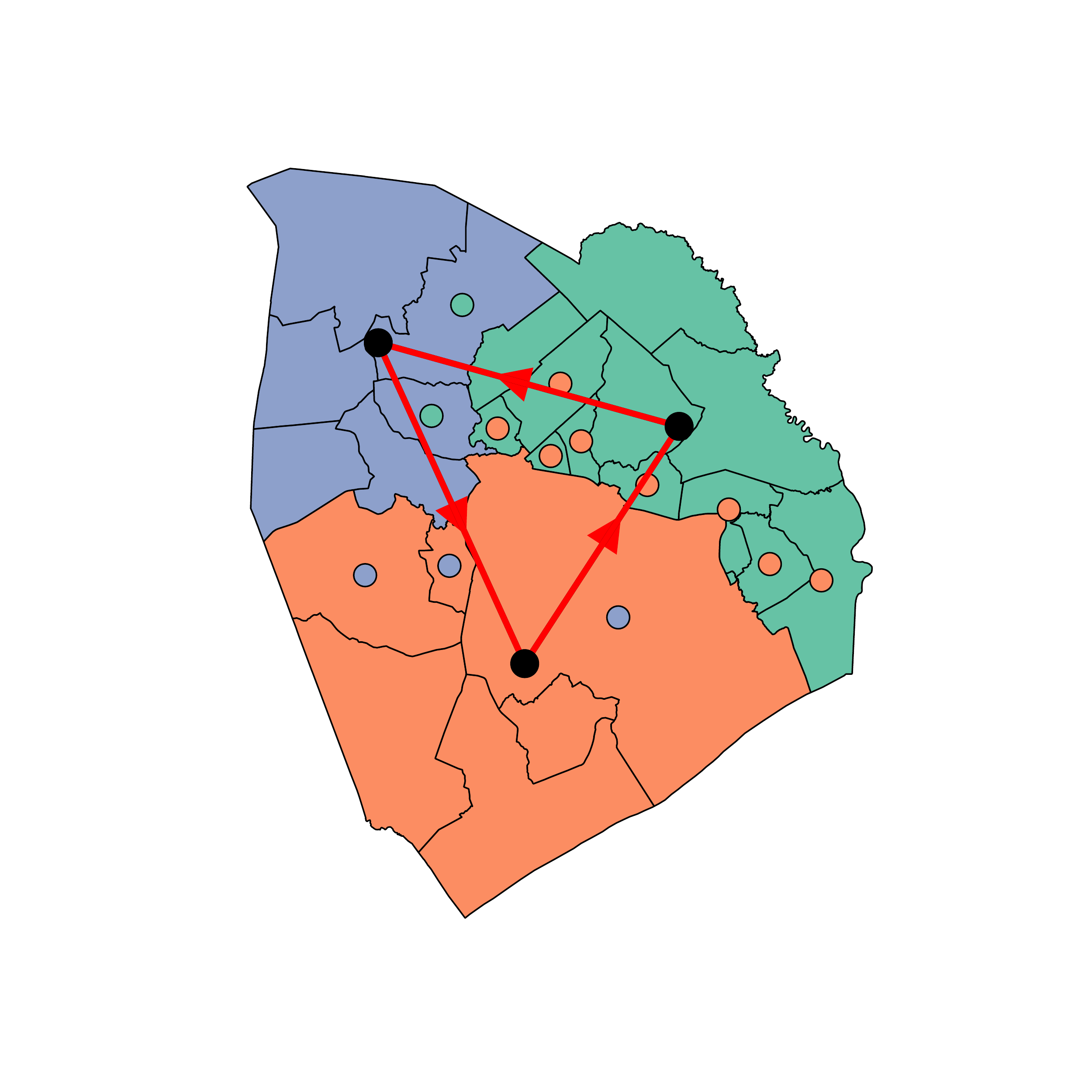}}
\caption{(A) Support of the district-to-district flow proposal distribution with $(\vel_{1,2},\vel_{1,3},\vel_{2,3}) =(1,1,1)$. The corresponding ordered district pairs $(1,2),(1,3), (2,3)$ are displayed as red, blue, and green colored directed edges of the district graph, respectively. (B) Support of a generalized version (see \cref{rem:genearlized:dgraph:flow}) of the district-to-district flow where a momentum variable $\vel_{(1,2,3)}$ is assigned to the red colored counter-clockwise oriented cycle on the district graph. The support is shown for  positive velocity $\vel_{(1,2,3,)}=1$. For both figures the same convention is used for displaying the support as in \cref{fig:sn}.
 }
\end{figure}

\subsubsection{Associated district graph}
If the number of district is three or larger, then, for certain redistricting plans we may have the situation that certain pairs of district do not share a border. That means that for a given redistricting plan $\dm\in \dmDomain$ the set of adjacent districts 
\[
\dedges(\dm) =  \left \{ (i,j)  \mid C_{e}^{+}(\dm) \cup C_{e}^{-}(\dm)  \neq \emptyset \right \},
\]
may be a proper subset of $\dedges$. We refer to the state dependent graph $\dgraph(\dm) = ( \dvertices, \dedges(\dm))$, where $\dvertices = \left \{1,\dots,\ndist \right\}$ is the index set of the districts, as the {\em district graph} of $\dm$; see \cref{subfig:dg:1}. If $e=(i,j) \in \dedges(\dm)$, we refer to $\xprop_{e}$ as an active kernel. The generic choice of weights as in \eqref{eq:weights:distr} ensures that only active kernels are selected in the proposal step.  

\begin{remark}
The total number of district pairs $\abs{\dedges}=\ndist(\ndist+1)/2$ scales quadratically in the number of districts $\ndist$, and thus keeping track of all entries in $\velvec$ may be memory intensive if $\ndist$ is large. Moreover, the stochastic process $(\dm_{n})_{n\in \mathbb{N}}$ typically converges to the target measure $\target$ (when accounting for internal symmetries/label permutation) before all proposal kernels have become active. For this reason one may not want to keep track for all velocities over the whole simulation time. Instead one may resample the velocity $\vel_{e}$ from the uniform measure on $\{-1,1\}$ whenever the kernel $\xprop_{e}$ becomes active after a transition which results in the previously non-adjacent district pairs $e=(i,j)$ to share a border.
\end{remark}

\begin{remark}[More general district-to-district flows / flows on district graph]\label{rem:genearlized:dgraph:flow}
\cref{alg:single:node:flip:distr} can be viewed as a special case of a much larger class of non-reversible algorithms which use the structure 
of the state dependent {\em district graph} $\dgraph(\dm) = ( \dvertices, \dedges(\dm))$. Instead of assigning velocities to edges as in  \cref{alg:single:node:flip:distr}, one may assign velocities to any type of subgraphs of $\dgraph(\dm)$, and associate proposals which mimic flows across the respective districts. For example, one may associate velocities with cycles of 3 districts (see \cref{subfig:dgf:2}), and choose the proposal $\xprop_{i,j,k}$ so that it mimics a district flow in clockwise/anti-clockwise direction depending on the value of the velocity $ \vel_{(i,j,k)}$, e.g.,
\begin{equation*}
\begin{aligned}
\xprop_{(i,j,k)}((\dm, \velvec), (\dm^{\prime}, \velvec^{\prime})) &= \mathbbm{1}_{\nb_{(i,j)}^{\vel}(\dm)}(\ccdot) \prop_{(i,j)}(\dm,\ccdot) +  \mathbbm{1}_{\nb_{(j,k)}^{\vel}(\dm)}(\ccdot) \prop_{(j,k)}(\dm,\ccdot)  \\
&\qquad+
  \mathbbm{1}_{\nb_{(k,i)}^{\vel}(\dm)}(\ccdot) \prop_{(k,i)}(\dm,\ccdot),
  \end{aligned}
\end{equation*}
with the terms on the right hand side as defined in the pair-wise district graph version of the algorithm. 
\end{remark}

\subsection{Computational aspects and choice of the temperature parameter}\label{subsec:comp:temp}
The computational efficiency of an MCMC scheme both depends on the cost per generated sample and the mixing rate of the Markov chain itself. It is intuitively clear that very high rejection rates will negatively affect the mixing speed of a Markov chain. This is in particular true in the case of non-reversible Markov chains constructed as presented here, since every rejection event will result in the reversal of a momentum variable, so that high rejection rates prevent kinetic like movement of the chain. Therefore, it may be worth to trade in higher computational costs for the generation of proposals / cost per step if rejection rates are reduced by that. In what follows we discuss this tradeoff in the case of two important special cases of the tempered proposal kernel of \cref{alg:single:node:flip:distr} (Note that the same observation follow as a special case for \cref{alg:single:node:flip:confl}.)
\begin{itemize}
\item If $\beta=1$, then the proposal distribution $\prop(\dm,\ccdot)$ is the target measure constrained to $\nb(\dm)$, and the Metropolis ratio simplifies to 
\begin{equation}\label{eq:mh-ratio:tempering}
r_{e}(\dm,\dm^{\prime}) =
\frac
{Z_{e}^{\vel_{e}}(\dm)}
{Z_{e}^{-\vel_{e}}(\dm^{\prime})}
\times
\frac
{Z_{e}(\dm^{\prime})}
{Z_{e}(\dm)}
\times
\frac
{Z(\dm) }
{Z(\dm^{\prime})}.
\end{equation}
In this case the acceptance probability of the proposal is not  dependent on the energy difference $\Delta J = J(\dm^{\prime}) - J(\dm)$, but is merely a function of various partition functions of the current state and the proposed state.
\item In the limit $\beta \rightarrow 0$,  the proposal distribution $\prop(\dm,\ccdot)$ becomes the uniform distribution on $\nb(\dm)$, and we thus recover for $\beta=0$ a version of the single node flip algorithm where the proposal is sampled uniformly from the set $\nb(\dm)$. In this case the acceptance probability becomes
 \begin{equation*}
r_{e}(\dm,\dm^{\prime}) = 
e^{-\Delta J } 
\times
\frac
{\abs{\nb_{e}^{\vel_{e}}(\dm)}}
{\abs{\nb_{e}^{-\vel_{e}}(\dm^{\prime})}}
\times
\frac
{\abs{\nb_{e}(\dm^{\prime})}}
{\abs{\nb_{e}(\dm)}}
\times
\frac
{\abs{\nb(\dm)} }
{\abs{\nb(\dm^{\prime})}}.
\end{equation*}
\end{itemize}
In terms of computational costs it is important to note, that in the case where $\beta>0$, the generation of a proposal and evaluation of the Metropolis-Hastings ratio requires the computation of the score function for all states in the neighborhood of $\dm$ as well as all states in the neighborhood of the proposal $\dm^{\prime}$. Depending on the form of the score function $J$ the computational costs for this operation may vary. In particular, if the form of $J$ is as in \eqref{eq:score:example}, computing the $J(\dm^{\prime})$ of a neighboring state of $\dm$ may only require reevaluation of the two terms in sums of the sub score functions which relate to the two districts which are modified in the transition from $\dm$ to $\dm^{\prime}$. 
In contrast to that the operations involved in generating a proposal and evaluating the Metropolis-Hastings ratio in the case of $\beta=0$ only require identifying the various neighborhoods of the current state $\dm$ and the proposal $\dm^{\prime}$.

The higher computational costs of using a tempered proposal with $\beta=1$, may be offset due to potentially drastically reduced rejection rates in the tempered case. For large redistricting maps the terms ${Z_{e}(\dm^{\prime})}/{Z_{e}(\dm)}$ and ${Z(\dm) }/{Z(\dm^{\prime})}$ in the Metropolis-Hastings ratio \eqref{eq:mh-ratio:tempering} can be expected to be with high probability close to 1. Similarly, the first factor may only be significantly smaller than 1 if extending the district $i$ into $j$ is in average energetically unfavorable in comparison to extending the $j$th district in to the $i$th district, a property which may give rise to cyclic alignment of district-to-district velocities; see \cref{sec:num}.

In contrast to that the Metropolis-Hastings ratio and thus the acceptance probability in the case $\beta=0$ does directly depend on the energy difference between the current state $\dm$ and $\dm^{\prime}$. In the situation where the variation of score values of states in the neighborhood of a given states is large this may result in drastically increased rejection rates in comparison the rejection rates in the tempered proposal. 

\begin{figure}
\begin{minipage}[t]{0.68\textwidth}
\begin{algorithm}[H]
\caption{Pair-wise district-to-district flow}\label{alg:single:node:flip:distr}
\SetKwInOut{Input}{input}
\SetKwInOut{Output}{output}
\Input{$\dm,\velvec$}
sample $e\sim \weight(\dm) =  [Z_{\tilde{e}}(\dm)/Z(\dm)]_{\tilde{e}\in \dedges} $\;
\eIf{$Z_{e}^{\vel_{e}}(\dm) = 0$}{
$\vel_{e} \gets -\vel_{e}$
}
{
sample $\dm^{\prime} \sim  \mathbbm{1}_{\nb_{e}^{\vel_{e}}(\dm)}(\ccdot) {\rm e}^{-\beta J(\ccdot)}$\;
sample $u \sim \mathcal{U}([0,1])$\;
\eIf{$Z_{e}^{-\vel_{e}}(\dm^{\prime})=0$ {\bf or} $u <
\frac{ 
e^{J(\dm )-J(\dm^{\prime})} Z_{e}^{\vel_{e}}(\dm)
}
{ 
 e^{\beta J(\dm) - \beta J(\dm^{\prime})}Z_{e}^{-\vel_{e}}(\dm^{\prime})
}
$}{
$\dm \gets \dm^{\prime}$\;
$\vel_{\tilde{e}} \sim \mathcal{U}(\{-1,1\}), \text{ for any edge } \tilde{e} \in \dedges(\dm^{\prime})\setminus \dedges(\dm)$\; 
}{
$\vel_{e} \gets -\vel_{e}$\;
}
}
\KwRet{$\dm,\velvec$}
\end{algorithm}
\end{minipage}
\caption*{
}
\end{figure}

\section{Numerical experiments}\label{sec:num}

To test our ideas numerically, we consider an example first presented in Section 4 of \cite{njatDedfordSolomon2019graphs} of a square lattice split into two districts.  As noted in the previous work, this districting problem is equivalent to considering a loop-free random walk that partitions the region. We consider the score function 
\begin{align}
J(\dm) &= J_{pop}(\dm) + J_C(\dm),
\end{align}
where the population score
\[
J_{\rm pop}(\dm) = \begin{cases} 0, & {\rm pop_{\min}}\leq  \abs{D_{1}(\dm)} \leq {\rm pop_{\max}} \\ \infty & \text{ otherwise}, \end{cases}
\]
is a hard constraint ensuring that the number of nodes in each district are between $ {\rm pop_{\min}}$ and  ${\rm pop_{\max}}$. The compactness score $J_{C}(\dm) = \abs{C(\dm)}$  corresponds to length of the boundary, i.e., the number of conflicted edges.

We sample the space of districting plans that are simply connected on a $10\times10$ square lattice with parameter values ${\rm pop_{min}}=45, {\rm pop_{max}}=55$  (i.e. we allow up to 10\% deviation from half of the lattice points).


We note that the score function $J(\dm)$ is minimized when the district boundary lies either perfectly horizontally or vertically, and decreases as the district boundary length grows. Intuitively we may think of four meta-stable states, one with a given district in the north, south, east or west, with an energetic barrier between these meta-stable states. Indeed, as demonstrated in \cite{njatDedfordSolomon2019graphs}, the single node flip algorithm may mix extremely slowly.  Using $J_C(\dm) = \log(10)|\text{cut}(\dm)|$ on a $40\times40$ lattice, the authors found that the system did not change meta-stable states even after nearly 3 billion steps when using the single node flip algorithm without tempering.

We evaluate the performance of the single node flip algorithm, non-reversible district-to-district flow algorithm, and non-reversible center of mass flow algorithm. For the single node flip algorithm we consider versions with tempering and without tempering. For the tempered algorithms we set $\beta = 0.5$. For the center of mass algorithm, we place a circular vortex field with an origin at the center of the redistricting graph; in Cartesian coordinates, the field is defined as $(\cos(\alpha), -\sin(\alpha))$, where $\alpha$ is the angle between the vector and the positive horizontal axis.

For each method we simulate 10 independent chains of $10^{7}$ samples. All chains are initialized in the same meta-stable state with one district in the north and south, respectively.  

We first evaluate the early mixing properties of the methods on basis of the first 25000 steps of a single chain for each method. For each vertex of the precinct graph, we compute the fraction $f$ of steps that vertex is assigned to district 1. By symmetry, the expected value of $f$ with respect to the target measure $\target$ is $1/2$ for every vertex.   
In order to accentuate deviations from the expected value, we examine the value
\begin{align}
\log(1+|f-1/2|) \sgn(f-1/2),
\end{align}
where $f$ is the fraction of time a node spends in a chosen district and $\EE[f] = 1/2$. We summarize this field variable in the early part of the chains in Figure~\ref{fig:fieldSummarySqLattice}. 

In the single node flip algorithm, we find that the plan is predominantly stuck in the north/south orientation (see Figure~\ref{fig:fieldSummarySqLattice}). Tempering helps to alleviate this issue, however, this method appears to favor diagonal cuts (or more likely oscillates between a vertical and horizontal orientation).  The district to district flow is biased toward a north/south districting plan, although apparently less so than the first two methods. The center of mass flow is nearly unbiased in its orientation. The largest deviation of $f$ from the expected value of any vertex is found to be $.45$, $.34$, $.23$, $.1$ for the single-node-flip algorithm, single-node-flip algorithm (tempered), district-to-district flow algorithm, and center-mass-flow algorithm, respectively.

We also examine the chains over all 10 million steps in the second row of Figure~\ref{fig:fieldSummarySqLattice}, and find that for all methods and all vertices $f$ does not deviate further than 0.02 from the expected value.

\begin{figure}
\begin{center}
\includegraphics[width=0.9\linewidth, clip = true, trim = {0cm 3cm 0cm 0cm}]{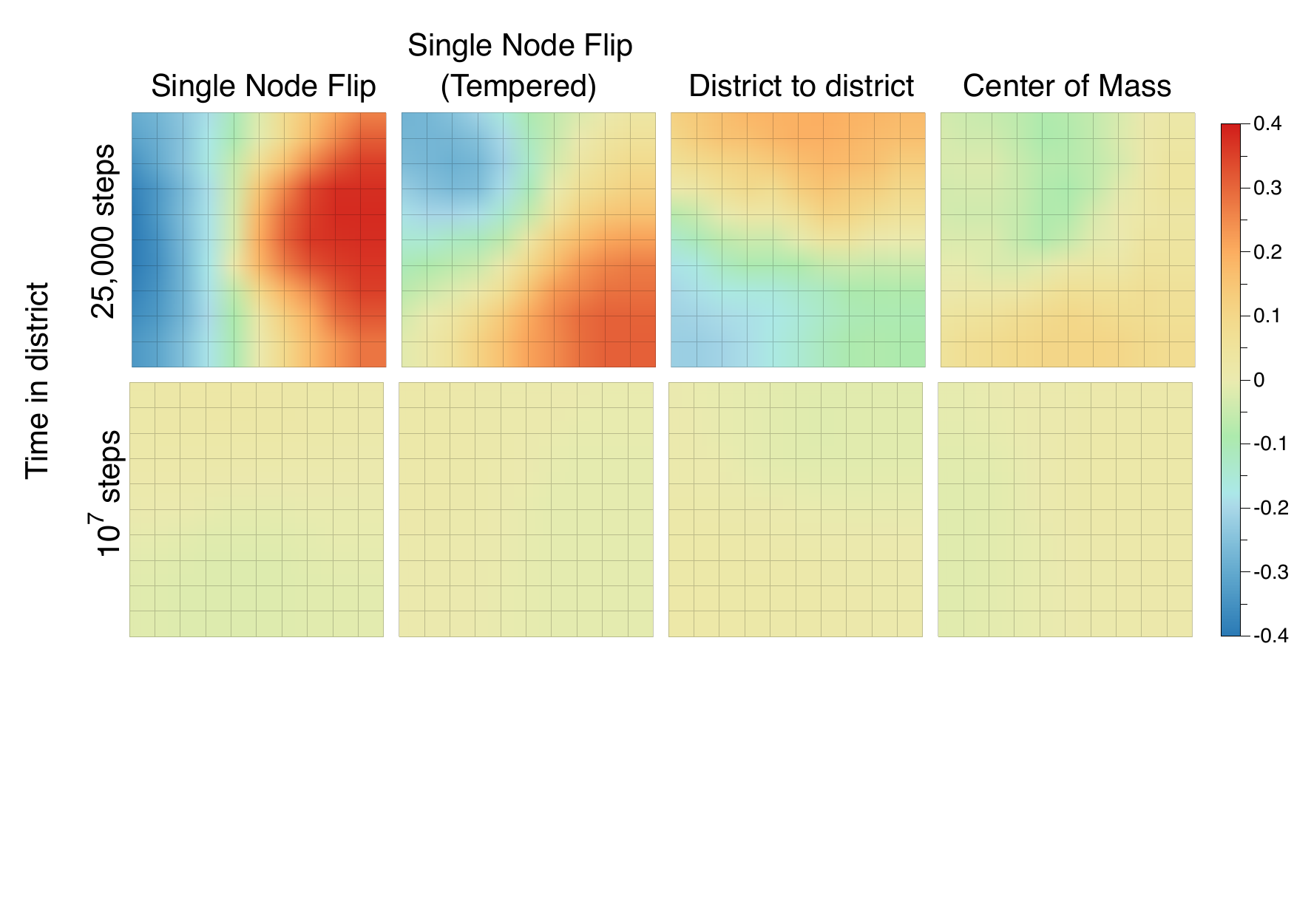}
\end{center}
\caption{We examine the fraction of time each node in a 10$\times$10 square lattice remains in a given district 
for four the methods after $25000$ steps (top row) and $10^7$ step (bottom row).  The methods are single node flip (left most), single node flip with tempered proposals (second from the left), center of mass flow (second from the right) and district to district flow (right most).}
\label{fig:fieldSummarySqLattice}
\end{figure}

To further probe the mixing rates of the chains, we examine how frequently the chain transitions between each meta-stable state. We define a meta stable state as a state where the square boundary between the two districts is within three nodes of a horizontal or vertical cut on an opposite side of the lattice.  We begin by examining the frequency that each chain spends in each of the meta-stable states after $10^7$ steps and look at the variance of these frequencies over the 10 chains (see Figure~\ref{fig:freqMSS}).  For the single node flip algorithm the range of frequencies for the meta-stable states varies between 3.2\% and 4.5\%; when adding tempering, the range of varied frequencies across chains lies between 3.6\% and 4.7\%; for district to district flow, the range lies between 2.8\% and 2.8\%.  In contrast, the center of mass flow frequencies range between 1.4\% and 2.8\%, meaning that the variance across runs is significantly lower than the other methods.

\begin{figure}
\begin{center}
\includegraphics[height=8cm, clip = true, trim = {0cm 0cm 0cm 0cm}]{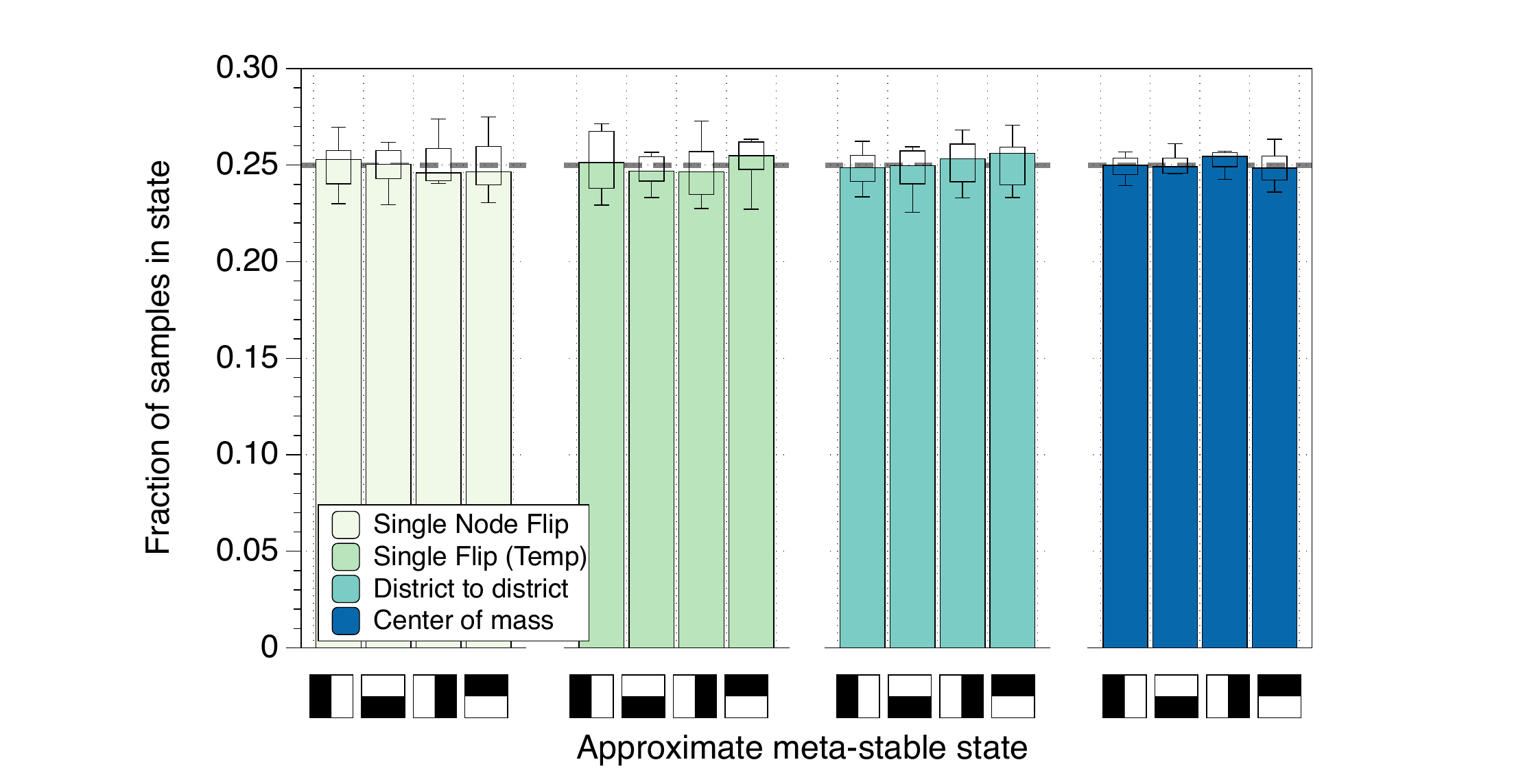}
\end{center}
\caption{We display the median frequency spent in the meta-stable states after $10^7$ proposals over 10 chains for each method.  We also display the deviation of these frequencies across the chains with standard box plots.  The half colored squares correspond to a given district being in the north, east, south and west (where shaded).  Dashed lines show the long time expected limit (i.e. equal probability of being in any of the meta-stable states))}
\label{fig:freqMSS}
\end{figure}

We expect that the reason for the faster mixing in the center of mass flow is due to a higher rate of transitions across meta-stable states.  We examine the frequency of transitions between meta-stable states in Figure~\ref{fig:freqTrans} after 10 millions steps on a single chain. We find that all of the methods are symmetric in terms of how often they transition between meta-stable states.  We find that the center of mass flows have significantly more transitions than all other methods -- this method transitions roughly 7,500 times in 10 million proposals which is more than twice the number of transitions of the other methods.  In contrast the single node flip tempered and district to district flow transition with a nearly identical frequency (roughly 3,600 times per 10 million proposals) and the single node flip method transitions even less (roughly 2,200 times per 10 million proposals).

\begin{figure}
\begin{center}
\includegraphics[height=8.5cm, clip = true, trim = {0cm 0cm 0cm 0cm}]{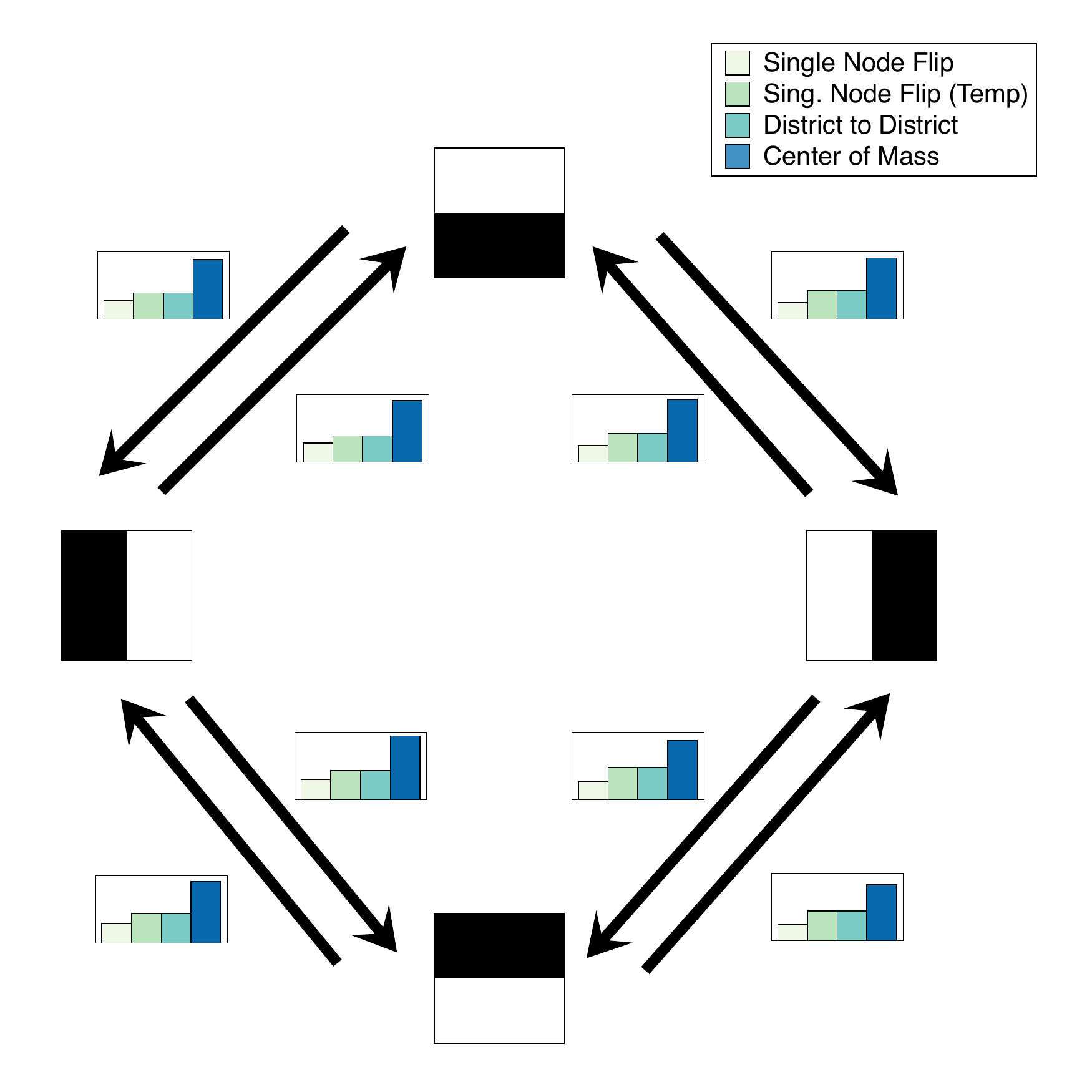}
\end{center}
\caption{We display the number of transitions between meta-stable states in the first 10 million proposals on a single chain for each of the four methods. The half colored squares correspond to a given district being in the north, east, south and west. All sub-figure bar charts have the same scale in the vertical axis with a range of zero to 1,100 transitions.}
\label{fig:freqTrans}
\end{figure}

We conclude by examining how the districting plan decorrelates along the trajectory of each sampling method.  To examine this correlation, we examine precinct assignment matrix, $\phi: \dmDomain \to \RR^{n \times d}$ such that $\phi(\dm)_{ij} = 1$ if precinct $i$ (of $n$ precincts) is assigned to district $j$ (of $d$ districts), $0$ else.  To compare two precinct assignments, we note that 
\begin{align}
\{\phi(\xi)\phi(\xi')^T\}_{ii} = \delta_{\xi(i)\xi'(i)},
\end{align}
which is to say that the diagonal of the product is one when the precinct is assigned to the same precinct and zero otherwise.  This implies that
\begin{align}
\frac{1}{n} \mathop{\mathrm{tr}}(\phi(\xi) \phi(\xi)^T) = 1.
\end{align}
Furthermore, by symmetry, we have that
\begin{align}
\EE(\phi(\xi))_{ij} = 1/d. 
\end{align}
We use the above facts to develop a measurement of precinct similarity given by
\begin{align}
G(t) = \frac{d}{n(d-1)} \EE \Bigg[ \mathop{\mathrm{tr}}\Big(\big(\phi(\dm_0) - \EE(\phi) \big) \big(\phi(\dm_{t})- \EE(\phi)\big)^T \Big) \Bigg],
\end{align}
where the outer expectation is taken with respect to the law of the process assuming $\dm_{0} \sim \target$. We remark that $G(0) = 1$ and that $G(t) \rightarrow 0$ as $t \rightarrow \infty$ since $\dm_{t}$ becomes independent of $\dm_0$ in the limit $t\rightarrow \infty$.  We also remark that this correlation is closely related to the evolution of the Hamming distance between two plans.  The Hamming distance, $d(\xi, \xi')$, counts the number of precincts that are assigned to different districts across two plans.  Thus
\begin{align}
\mathop{\mathrm{tr}}(\phi(\xi) \phi(\xi')^T) = n - d(\xi, \xi').
\end{align}

We estimate $G(t)$ by taking 100,000 boot-strapped samples from each chain of 10 million samples.  From each sample, we gather statistics as a function of progressive steps from the sample and then average across all chains.  We plot $G(t)$ as a function of the steps from the initial samples in Figure~\ref{fig:autocor}. We find that all four methods decorrelate within 50,000 steps.  We find that the center of mass flow decorrelates significantly faster than the three other methods. The district to district and single node flip tempered algorithms decorrelate at nearly the same rate, which makes sense due to the fact that a two district state does not admit cycles which can create longer time flows without rejection.  The single node flip algorithm decorrelates far slower than all other methods.
\begin{figure}
\begin{center}
\includegraphics[width=0.9\linewidth, trim = {1.25cm 0cm 0cm 0cm}]{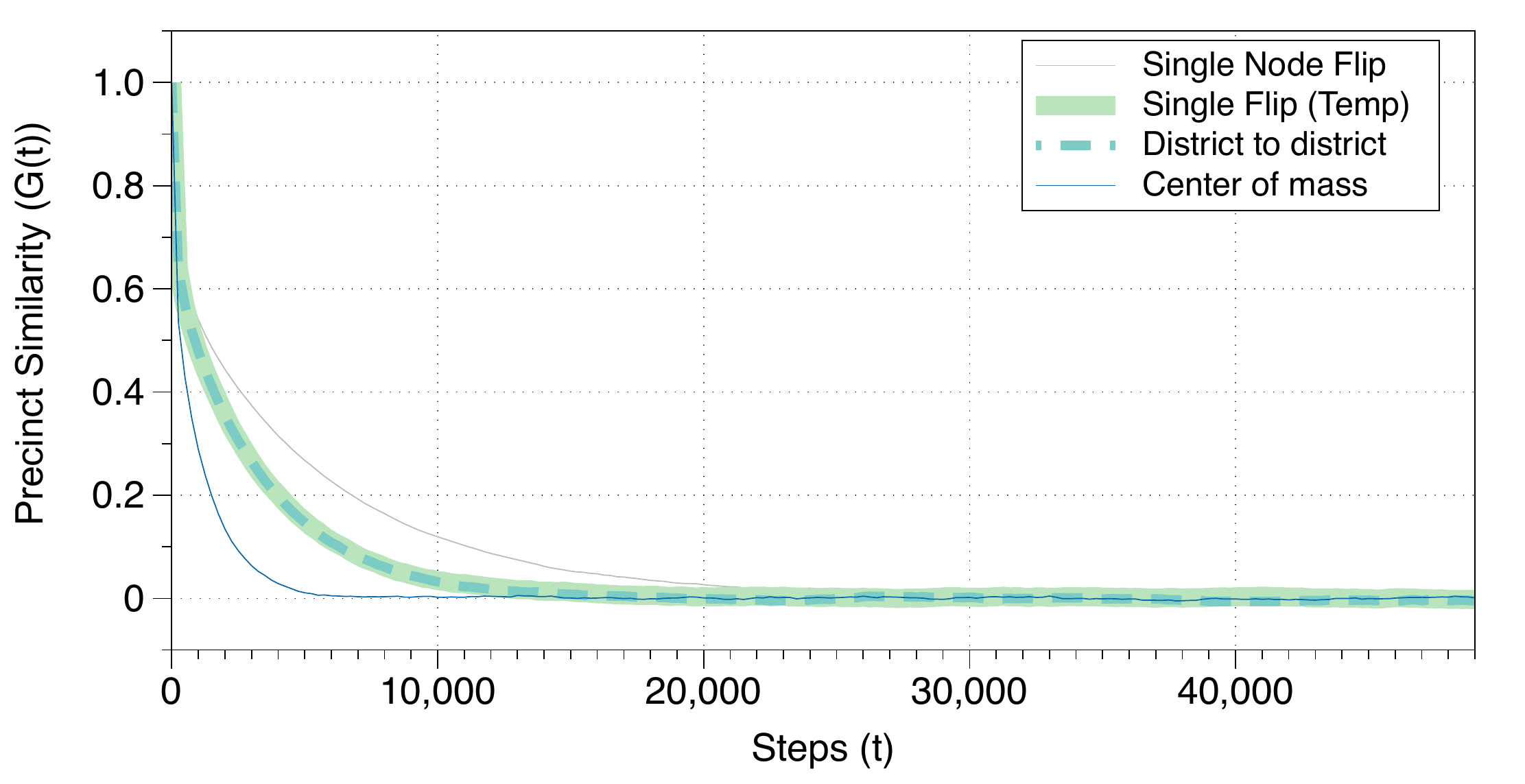}
\end{center}
\caption{We show the correlation of precinct assignments, $G(t)$, as a function of the number of steps taken for each of the four methods.  The results are taken by bootstrap sampling over the 10 different runs on the first $10^7-5\times10^4$ steps for each method and then gathering statistics for $G(t)$ starting with each of the samples.}
\label{fig:autocor}
\end{figure}

\section*{Discussion}
To the best of our knowledge this article is the first work to propose using non-reversible MCMC schemes in the context of sampling of redistricting plans. We explore various natural choices of introducing non-reversibility in this application (i.e., via flows induced by a vector field, and district-to-district flows).  We provide the necessary (and novel) mathematical framework for implementing these approaches in the form of the Mixed Skew Metropolis-Hastings algorithm, which relies on a generalization of the skew detailed balance condition involving multiple momentum and a mixture of proposal kernels. The resulting Markov chain, on a state space extended to include unit momentum, satisfies mixed skew detailed balance. In this setting, we derive conditions which ensure ergodicity of the generated Markov chain with respect to the target measure.

This framework can be used to derive many instances of non-reversible MCMC algorithms, both in general and for the sampling of redistricting plans. We focus in this article on two specific simple implementations of the above mentioned approaches, namely the center-of-mass flow (\cref{alg:vf}) and the pair-wise district-to-district flow (\cref{alg:single:node:flip:distr}). Even with this simple implementations, we find in numerical experiments that the proposed methods mix significantly better than comparable reversible methods. We expect that further improvement can be achieved by more principled implementations of the here suggested approaches. This may open many interesting directions of research. This pertains in particular the choice of the vector field $\vf$, the embedding $\embedding$ and the orientation function. Similarly, we expect that a more refined assignment of cycles on the district graph as alluded to in \cref{rem:genearlized:dgraph:flow} may lead to significant improvement in performance in redistricting problems with a larger number of districts. 

\section*{Acknowledgement}
Matthias Sachs acknowledges support from grant DMS-1638521 from SAMSI.  Jonathan Mattingly  acknowledges  the partial support the NSF grant DMS-1613337. Jonathan Mattingly and Gregory Herschlag acknowledge partially support from the Duke Mathematics Department, the Rhodes Information initiative at Duke, and the Duke Dean's and Provost's offices. All authors are thankful for the support of SAMSI and the Duke NSF-TRIPOS  grant  (NSF-CFF-1934964) which have supported activity around the mathematics of Redistricting during the time of this work which was very stimulating and supportive. We also acknowledge the PRUV and Data+ undergraduate research programs under which our initial work in Quantifying Gerrymander was developed staring in 2013. Jonathan Mattingly is thankful for very educational and inspiring  conversations with Manon Michel while at CIRM in Marseille in September 2018 which started him on this line of inquire. We would also like to thank Andrea Agazzi, Jianfeng Lu, and the Triangle Quantifying Gerrymandering Meet-up for useful and stimulating discussions.

\bibliographystyle{abbrv}
\bibliography{refs_gerry.bib}

\begin{thebibliography}{10}

\bibitem{akhmatskaya2009comparison}
E.~Akhmatskaya, N.~Bou-Rabee, and S.~Reich.
\newblock A comparison of generalized hybrid {M}onte {C}arlo methods with and
  without momentum flip.
\newblock {\em Journal of Computational Physics}, 228(6):2256--2265, 2009.

\bibitem{bangia2017redistricting}
S.~Bangia, C.~V. Graves, G.~Herschlag, H.~S. Kang, J.~Luo, J.~C. Mattingly, and
  R.~Ravier.
\newblock Redistricting: Drawing the line.
\newblock {\em arXiv preprint arXiv:1704.03360}, 2017.

\bibitem{bierkens2019zig}
J.~Bierkens, P.~Fearnhead, G.~Roberts, et~al.
\newblock The zig-zag process and super-efficient sampling for {B}ayesian
  analysis of big data.
\newblock {\em The Annals of Statistics}, 47(3):1288--1320, 2019.

\bibitem{bouchard2018bouncy}
A.~Bouchard-C{\^o}t{\'e}, S.~J. Vollmer, and A.~Doucet.
\newblock The bouncy particle sampler: A nonreversible rejection-free {M}arkov
  chain {M}onte {C}arlo method.
\newblock {\em Journal of the American Statistical Association},
  113(522):855--867, 2018.

\bibitem{carter2019merge}
D.~Carter, G.~Herschlag, Z.~Hunter, and J.~Mattingly.
\newblock A merge-split proposal for reversible {M}onte {C}arlo {M}arkov chain
  sampling of redistricting plans.
\newblock {\em arXiv preprint arXiv:1911.01503}, 2019.

\bibitem{deford2019recombination}
D.~DeFord, M.~Duchin, and J.~Solomon.
\newblock Recombination: A family of {M}arkov chains for redistricting.
\newblock {\em arXiv preprint arXiv:1911.05725}, 2019.

\bibitem{diaconis2000analysis}
P.~Diaconis, S.~Holmes, and R.~M. Neal.
\newblock Analysis of a nonreversible {M}arkov chain sampler.
\newblock {\em Annals of Applied Probability}, pages 726--752, 2000.

\bibitem{dobson2020reversible}
P.~Dobson, I.~Fursov, G.~Lord, and M.~Ottobre.
\newblock Reversible and non-reversible {M}arkov chain {M}onte {C}arlo
  algorithms for reservoir simulation problems.
\newblock {\em Computational Geosciences}, pages 1--13, 2020.

\bibitem{duncan2017nonreversible}
A.~Duncan, G.~Pavliotis, and K.~Zygalakis.
\newblock Nonreversible {L}angevin samplers: Splitting schemes, analysis and
  implementation.
\newblock {\em arXiv preprint arXiv:1701.04247}, 2017.

\bibitem{duncan2016variance}
A.~B. Duncan, T.~Lelievre, and G.~Pavliotis.
\newblock Variance reduction using nonreversible {L}angevin samplers.
\newblock {\em Journal of statistical physics}, 163(3):457--491, 2016.

\bibitem{fifield2015new}
B.~Fifield, M.~Higgins, K.~Imai, and A.~Tarr.
\newblock A new automated redistricting simulator using {M}arkov chain {M}onte
  {C}arlo.
\newblock {\em Work. Pap., Princeton Univ., Princeton, NJ}, 2015.

\bibitem{herschlag2018quantifying}
G.~Herschlag, H.~S. Kang, J.~Luo, C.~V. Graves, S.~Bangia, R.~Ravier, and J.~C.
  Mattingly.
\newblock Quantifying {G}errymandering in {N}orth {C}arolina.
\newblock {\em arXiv preprint arXiv:1801.03783}, 2018.

\bibitem{herschlag2017evaluating}
G.~Herschlag, R.~Ravier, and J.~C. Mattingly.
\newblock Evaluating partisan {G}errymandering in {W}isconsin.
\newblock {\em arXiv preprint arXiv:1709.01596}, 2017.

\bibitem{hukushima2013irreversible}
K.~Hukushima and Y.~Sakai.
\newblock An irreversible {M}arkov-chain {M}onte {C}arlo method with skew
  detailed balance conditions.
\newblock In {\em Journal of Physics: Conference Series}, volume 473, page
  012012. IOP Publishing, 2013.

\bibitem{ma2019irreversible}
Y.-A. Ma, E.~B. Fox, T.~Chen, and L.~Wu.
\newblock Irreversible samplers from jump and continuous {M}arkov processes.
\newblock {\em Statistics and Computing}, 29(1):177--202, 2019.

\bibitem{jcmReport}
J.~C. Mattingly.
\newblock Expert report for {C}ommon {C}ause v. {L}ewis.
\newblock {\em {C}ommon {C}ause v. {L}ewis}, 2019.

\bibitem{jcmRebuttal}
J.~C. Mattingly.
\newblock Rebuttal of defendant's expert reports for {C}ommon {C}ause v.
  {L}ewis.
\newblock {\em Common Cause v. Lewis}, 2019.

\bibitem{mattingly2014redistricting}
J.~C. Mattingly and C.~Vaughn.
\newblock Redistricting and the will of the people, 2014.

\bibitem{michel16:_irrevMCMC}
M.~Michel.
\newblock {\em Irreversible {M}arkov chains by the factorized {M}etropolis
  filter: Algorithms and applications in particle systems and spin models}.
\newblock PhD thesis, {\'E}cole {N}ormale {S}upe\'erieure, 2016.

\bibitem{najt2019complexity}
L.~Najt, D.~DeFord, and J.~Solomon.
\newblock Complexity and geometry of sampling connected graph partitions, 2019.

\bibitem{njatDedfordSolomon2019graphs}
L.~Najt, D.~Deford, and J.~Solomon.
\newblock Complexity and geometry of sampling connected graph partitions.
\newblock {\em Preprint; arxiv.org/pdf/1908.08881.pdf}, 2019.

\bibitem{neal2004improving}
R.~M. Neal.
\newblock Improving asymptotic variance of {MCMC} estimators: Non-reversible
  chains are better.
\newblock {\em arXiv preprint math/0407281}, 2004.

\bibitem{ottobre2016function}
M.~Ottobre, N.~S. Pillai, F.~J. Pinski, A.~M. Stuart, et~al.
\newblock A function space {HMC} algorithm with second order {L}angevin
  diffusion limit.
\newblock {\em Bernoulli}, 22(1):60--106, 2016.

\bibitem{stoltz2010free}
G.~Stoltz, M.~Rousset, et~al.
\newblock {\em Free energy computations: A mathematical perspective}.
\newblock World Scientific, 2010.

\bibitem{sun2010improving}
Y.~Sun, J.~Schmidhuber, and F.~J. Gomez.
\newblock Improving the asymptotic performance of {M}arkov chain
  {M}onte-{C}arlo by inserting vortices.
\newblock In {\em Advances in Neural Information Processing Systems}, pages
  2235--2243, 2010.

\bibitem{suwa2012general}
H.~Suwa and S.~Todo.
\newblock General construction of irreversible kernel in {M}arkov chain {M}onte
  {C}arlo.
\newblock {\em arXiv preprint arXiv:1207.0258}, 2012.

\bibitem{Tierney_1998}
L.~Tierney.
\newblock A note on {M}etropolis-{H}astings kernels for general state spaces.
\newblock {\em The Annals of Applied Probability}, 8(1):1–9, Feb 1998.

\bibitem{turitsyn2011irreversible}
K.~S. Turitsyn, M.~Chertkov, and M.~Vucelja.
\newblock Irreversible {M}onte {C}arlo algorithms for efficient sampling.
\newblock {\em Physica D: Nonlinear Phenomena}, 240(4-5):410--414, 2011.

\bibitem{vucelja2016lifting}
M.~Vucelja.
\newblock Lifting---a nonreversible {M}arkov chain {M}onte {C}arlo algorithm.
\newblock {\em American Journal of Physics}, 84(12):958--968, 2016.

\end{thebibliography}

\appendix 
\section{Algorithms}
\subsection{Tempered single node-flip algorithm}
Combining the tempered single node-flip proposal of \cref{sec:SingleNodeFlip} with a Metropolis-Hastings accept-reject step results in the following algorithm.\\
\begin{figure}[h]
\begin{minipage}[t]{0.67\textwidth}
\begin{algorithm}[H]
\caption{Tempered single node-flip algorithm.}\label{alg:single:node:flip:confl}
\SetKwInOut{Input}{input}
\SetKwInOut{Output}{output}
\Input{$\dm$}
generate proposal $\dm' \sim \mathbbm{1}_{\nb(\dm)}(\ccdot) {\rm e}^{-\beta J(\ccdot)}$\;
$r(\dm,\dm^{\prime}) 
= \dfrac{ 
e^{-J(\dm^{\prime})+\beta J(\dm^{\prime})}Z(\dm) 
}
{ 
e^{-J(\dm )+\beta J(\dm)}Z(\dm^{\prime})
}
$\;
sample $u \sim \mathcal{U}([0,1])$\;
\If{$u < r(\dm,\dm^{\prime})$}{
$\dm \gets \dm^{\prime}$
}
\KwRet{$\dm$}
\end{algorithm}
\end{minipage}
\end{figure}

Note that for
\begin{itemize}
\item $\beta=1$, then the proposal distribution $\prop(\dm,\ccdot)$ is the target measure constrained to $\nb(\dm)$, and the Metropolis ratio simplifies to 
\[
r(\dm,\dm^{\prime}) =
\frac
{Z(\dm)}
{Z(\dm^{\prime})}.
\]
In this case the acceptance probability of the proposal is not  explicitly dependent on the energy difference $\Delta J = J(\dm^{\prime}) - J(\dm)$, but is merely a function of the partition functions of the current state and the proposed state.
\item $\beta=0$ a version of the single node flip algorithm where the proposal is sampled uniformly from the set $\nb(\dm)$ is recovered and the Metropolis-Hastings ratio becomes $r(\dm,\dm^{\prime}) = \frac{ 
e^{-J(\dm^{\prime})} 
\abs*{\nb(\dm^{\prime})}^{-1}
}
{ 
e^{-J(\dm )}
\abs*{ \nb(\dm)}^{-1} 
}
$, tends to be lower in comparison to the tempered proposal with $\beta=1$.
\end{itemize}
Furthermore, if the underlying tempered proposal generates an irreducible Markov chain, then the tempered MCMC algorithm is also irreducible and hence has the desired target measure as its unique invariant measure. This is summarized in the following proposition.
\begin{proposition}
If the topology of the graph $\pgraph=(\pvertices,\pedges)$ is such that the Markov chain generated by the single-node-flip proposal as defined in \eqref{eq:prop:sn} is irreducible, then 
the probability measure $\target$ is the unique invariant measure of the Markov chain $(\dm_{n})_{n\in \NN}$ generated by the MCMC algorithm given in \cref{alg:single:node:flip:confl}.
\end{proposition}

\section{Proofs}\label{sec:proofs}
\subsection{Proof of \cref{thm:MSDB}   }We now return to the proof of \cref{thm:MSDB} which guaranteed that the mixed skewed balance condition ensures that the desired target measure $\pi$ is an invariant measure for the Markov chain  $\xtransc$.  The proof has the same structure as the analogous proof for the original  skewed balance condition  \cite{diaconis2000analysis}.
\begin{proof}[Proof of \cref{thm:MSDB}]
\begin{align}
 (\xtarget\xtransc)(\x)= \sum_{\x^{\prime}\in\xDomain} \xtarget(\x^{\prime}) \xtransc(\x^{\prime},\x) 
&= 
\sum_{\x^{\prime}\in\xDomain}\sum_{i=1}^{n}\xtarget(\x^{\prime}) \xweight_{i}(\x^{\prime}) \xtrans_{i}(\x^{\prime},\x) \notag \\
&=
\sum_{\x^{\prime}\in\xDomain}\sum_{i=1}^{n} \xweight_{i}(\x)  \xtarget(\x)  \xtrans_{i}(S_{i}(\x),S_{i}(\x^{\prime}))\notag \\
 &= \notag
 \xtarget(\x)\sum_{i=1}^{n}\xweight_{i}(\x) \Big(\sum_{\x^{\prime}\in\xDomain}\xtrans_{i}(S_{i}(\x),S_{i}(\x^{\prime})) \Big)\\
 &= \notag
 \xtarget(\x)\Big(\sum_{i=1}^{n}\xweight_{i}(\x) \Big)\\
&=
\xtarget(\x), \notag 
\end{align}
where the third equality follows by the mixed skew detailed balance condition \eqref{eq:generalized:db:mod},  
the fifth and sixth equality follow due to the fact that the sum within each pair of parentheses sums to one. In the first case because   $\xtrans_{i}$ is a Markov  transition kernel and $S_{i}(\xDomain)=\xDomain$ and in the second case because    because weights $\xweight_{i}(\x)$ sum up to one.
\end{proof}

\subsection{Proof of \cref{thm:MSMH} }We now give the proof that the Mixed Skew Metropolis-Hastings (MSMH) algorithm given in \cref{alg:generalized:MH} satisfies the mixed skew detailed balance condition, given in \eqref{eq:generalized:db:mod}; and hence, has the desired target measure $\pi$ as an invariant measure.  
\begin{proof}[Proof of  \cref{thm:MSMH}]
The form of the algorithm directly implies that the transition matrices/kernels $\xtrans_{i}, i=1,\dots,n$ of the generated Markov chain are of the form
\begin{equation*}
\xtrans_{i}(\x, \x^{\prime}) = g_{i,1}(\x, \x^{\prime}) + g_{i,2}(\x, \x^{\prime}),
\end{equation*}
where
\begin{equation}\label{eq:def:g}
\begin{aligned}
 g_{i,1}(\x, \x^{\prime})  &= \xprop_{i}(\x, \x^{\prime}) \min(1,r_{i}(\x, \x^{\prime}) ), \\
 g_{i,2}(\x, \x^{\prime})  &=  \mathbbm{1}_{\{S_{i}(\x)\}}(\x^{\prime}) \Big[ 1- \sum_{\tilde{\x}\in \xDomain} \xprop_{i}(\x, \tilde{\x})\min(1,r_{i}(\x, \tilde{\x}) ) \Big],
\end{aligned}
\end{equation}
thus a generalized skew detailed balance condition is satisfied for $\xtransc=\xweight \cdot \trans$, if 
\begin{equation}\label{eq:cond:gi2}
g_{i,j}(\x, \x^{\prime}) \pi(\x) \weight_{i}(\x)= g_{i,j} \left (S_{i}(\x^{\prime}),S_{i}(\x) \right) \pi(\x^{\prime})\weight_{i}(\x^{\prime})
\end{equation}
for $j=1,2$ and all  $\x,\x^{\prime} \in \xDomain$ and all $i=1,\dots,n$. We will now verify this condition for each $j$ separately.
\begin{itemize}[wide]
\item[Case $j=1:$]
A simple computation shows that 
\begin{equation*}
r_{i}(\x,\x^{\prime})  = \frac{1}{r_{i}(S_{i}(\x^{\prime}),S_{i}(\x))}.
\end{equation*}
Thus, $r_{i}(\x,\x^{\prime}) \in (0,1) \iff r_{i}(S_{i}(\x^{\prime}),S_{i}(\x))>1$, and  $r_{i}(\x,\x^{\prime}) = r_{i}(S_{i}(\x^{\prime}),S_{i}(\x)) \iff r_{i}(\x,\x^{\prime}) = 1$.
Therefore, if $r_{i}(\x,\x^{\prime})\leq 1$, we find
\begin{equation*}
\begin{aligned}
 g_{i,1}(\x, \x^{\prime}) \xtarget(\x) &= \min(1,r_{i}(\x, \x^{\prime}) )\prop(\x, \x^{\prime}) \xtarget(\x)\\
 &=   \frac{\prop(S_{i}(\x^{\prime}),S_{i}(\x)) \xtarget(\x^{\prime}) }{\prop(\x, \x^{\prime})\xtarget(\x)} \prop(\x,\x^{\prime})\xtarget(\x)\\
 &= g_{i,1}\left (S_{i}(\x^{\prime}),S_{i}(\x) \right) \xtarget(\x^{\prime}),
 \end{aligned}
\end{equation*}
and, similarly, $r_{i}(\x,\x^{\prime}) > 1 \iff  r_{i}(S_{i}(\x^{\prime}),S_{i}(\x)) \in (0,1)$, thus
\begin{equation*}
\begin{aligned}
 g_{i,1}(\x, \x^{\prime}) \xtarget(\x) 
 &= \prop(\x, \x^{\prime}) \min(1,r_{i}(\x, \x^{\prime}) ) \xtarget(\x)\\
 &=\prop(\x, \x^{\prime}) \xtarget(\x)\\
 &= 
 \frac{
 \prop(\x,\x^{\prime})\xtarget(\x)
 }
 {
 \prop(S_{i}(\x^{\prime}),S_{i}(\x)) \xtarget(\x^{\prime})
 }
 \prop\left(S_{i}(\x^{\prime}),S_{i}(\x) \right) \xtarget(\x^{\prime})\\
 &=\min\left(1,r_{i}(S_{i}(\x^{\prime}),S_{i}(\x)) \right)\, \prop\left(S_{i}(\x^{\prime}),S_{i}(\x) \right) \xtarget(\x^{\prime}) \\
 &= g_{i,1}\left (S_{i}(\x^{\prime}),S_{i}(\x) \right) \xtarget(\x^{\prime}).
 \end{aligned}
\end{equation*}
\item[Case $j=2:$]
For $S_{i}(\x) \neq \x^{\prime}$, it follows from the definition of $g_{i,2}(\x, \x^{\prime})$ in \eqref{eq:def:g} that  $ g_{i,2}(\x, \x^{\prime}) =   g_{i,2} \left (S_{i}(\x^{\prime}),S_{i}(\x) \right)=0$ and   \eqref{eq:cond:gi2} is thus trivially satisfied. If $S_{i}(\x) = \x^{\prime}$, it follows that
\begin{equation*}
g_{i,2}(\x, \x^{\prime}) = g_{i,2}(S_{i}^{2}(\x), S_{i}^{2}(\x^{\prime})) =  g_{i,2}(S_{i}(\x^{\prime}), S_{i}(\x)),
\end{equation*}
which by virtue of the fact that $\xtarget$ is invariant under $S_{i}$ also implies \eqref{eq:cond:gi2}.
\end{itemize}
\end{proof}

\subsection{Proof of \cref{thm:ergodic:general}}We now turn  to the proof of  the unique ergodicity and convergence of averages result given in \cref{thm:ergodic:general}.

Since  $\xtarget$ is invariant under $\xtransc$, that is $\xtarget\xtransc=\xtarget$, and $\xtarget(\dm)>0$ for all $\dm\in \dmDomain,$ it is sufficient to show that the Markov chain generated by $\xtransc$ is irreducible, meaning that we can reach any extended state $(\dm',\velvec')$ from any other extended state $(\dm,\velvec)$ within a finite number of steps, i.e.,
\begin{equation}\label{eq:xtransc:irred}
\forall\, (\dm,\velvec), (\dm',\velvec') \in \dmDomain \times \{-1,1\}^{n},\; \exists\, m\in \NN, \quad \text{s.t.}\quad \xtransc^{m}((\dm,\velvec),(\dm',\velvec')) >0.
\end{equation}


We will prove this statement by combining the following two intermediate facts.
\begin{enumerate}
\item Show that we can find a path with non-zero probability between any state  $(\dm,\velvec) \in  \dmDomain \times \{-1,1\}^{n}$ in the set $\dm' \times  \{-1,1\}^{n}$ sitting above an arbitrary point $\xi' \in \dmDomain $.
(\cref{col:trans} below). Such a path will be used to move about the statespace.
\item Show that for any $(\dm,\velvec) \in  \dmDomain \times \{-1,1\}^{n}$, if $i\in\activeset(\dm)$, we can find a path with non-zero probability beginning at $(\dm,\velvec)$ and ending at $(\dm,R_{i}(\velvec))$. These paths will be used to flip the $i$th momentum as needed, without changing the underlining state $\dm$.
(\cref{col:flippable:2} below).
\end{enumerate}

The basic idea of the proof of \cref{thm:ergodic:general}   is that the first of the above fact allows one to move around the state space from $\dm$ to $\dm'$ but without control of what happens to the momentum variable $\velvec$. The second fact above then allows one to modify the resulting $\velvec''$ to the desired $\velvec'$. We will see that at its core, the first fact will follow from the irreducibility of the proposal kernel $\prop$ which was an input to our algorithm. The second fact will follow again from this base irreducibility along with \cref{as:prop:ir} (or in light of \cref{lem:flippable}, \cref{as:circuit}) to ensure there is a place along the path to flip the signs in the momentum $\velvec''$ until it agrees with  $\velvec'$. This is the basic arc of the proof of  \cref{thm:ergodic:general}, though each of the above statements will need a little refinement and a few additional technical elements will need to be added to complete the argument. A graphic representation of the sketch of the proof is given in \cref{fig:proofoutline}.

We begin by stating the two lemma which correspond to the two above statements. The proofs of these lemma are postponed to the end of the section. 
\begin{lemma}\label{col:trans}
  For all extended states $(\dm,\velvec)\in \dmDomain \times \{-1,1\}^{n}$ and any state $\dm' \in \dmDomain$ there exists $m\in \NN$ so that 
\begin{equation*}
  \xtransc^{m} ( (\dm,\velvec), \{ \dm'\} \times \{-1,1\}^{n} ) >0. 
\end{equation*}
Additionally, $m$ can be chosen so that there is a path $((\dm_0,\velvec_0),\cdots,  (\dm_m,\velvec_m))$, of positive probability starting at $(\dm,\velvec)$ and ending in the set $\{ \dm'\} \times \{-1,1\}^{n}$, satisfying the following property: for every $i \in \{1,\dots,n\}$ there exist a $k \in \{0,\dots,m\}$ (depending on $i$) with $i \in \mathcal{A}(\dm_k)$. 
\end{lemma}

\begin{lemma}\label{col:flippable:2}
For any extended state $(\dm,\velvec) \in \dmDomain \times \{-1,1\}^{n}$ and $i\in \activeset(\dm)$, there is $m\in \NN$ so that 
\[
\xtransc^{m}\left (\big (\dm,\velvec \big), \big (\dm,R_{i}(\velvec) \big ) \right ) >0. 
\]
\end{lemma}

By the construction of the MSMH chain $\xtransc$ from the proposal chain $\prop$ in Section~\ref{ssec:SkewDetailedBalanceonStateGraph}, in particular the structure $\nb_{i}^{+}$ and $\nb_{i}^{-}$, imply that  
\begin{equation}
  \label{eq:suppQP}
  \bigcup_{i \in \activeset(\dm)} \Big[\supp \xtrans_{i}\big((\dm,\velvec),(\ccdot,\velvec) \big) \cup \supp \xtrans_{i}\big((\dm,R_{i}(\velvec)),(\ccdot,R_{i}(\velvec))\big)\Big]  = {\rm supp} \, \prop(\dm,\ccdot)\,.
\end{equation}
The one central obstacle in deducing \cref{col:trans} from this fact is that the particular coordinate momentum $\velvec$ might not be aligned in the right direction to allow us to propose and then follow a given transition which is possible under $\prop$. \cref{col:flippable:2} allows us to flip particular coordinate momentum $\vel_i$ provided $i$ is active. One remaining concern is that we might be forced to constantly correct the momentum which are changed as a side-effect of previous moves. The following lemma is central to ruling out this scenario.
\begin{lemma}\label{lem:trans}
Let \cref{as:prop} hold and consider the transition kernel $\xtransc((\dm,\velvec),\cdot ) = \sum_{i=1}^{n} \weight_{i}(\dm) \xtrans_{i}((\dm,\velvec),\ccdot)$ of the Markov chain generated by \cref{alg:generalized:MHb} with generic weights as specified in \eqref{eq:weights:generic}. 
%
For any  two states $\dm,\dm'\in \dmDomain$ and an  $i \in \{1,\dots,n\}$, such that $i\in \activeset(\dm)$ and $i \not\in \activeset(\dm')$, the following conclusion holds: 
\[
\forall \velvec \in \{-1,1\}^{n},\quad P_{i}((\dm,\velvec),(\dm',\velvec)) = 0. 
\]
\end{lemma}
Equipped now with \cref{col:trans}, \cref{col:flippable:2} and \cref{lem:trans} as well \eqref{eq:suppQP}, we are able of give the proof of \cref{thm:ergodic:general}.  As already mentioned, the basic outline is given in \cref{fig:proofoutline}. The proofs of these lemmas are postponed to the end of this section.
\begin{figure}
\begin{center}
\includegraphics[width=\textwidth, clip=true, trim={0 10cm 10cm 0}]{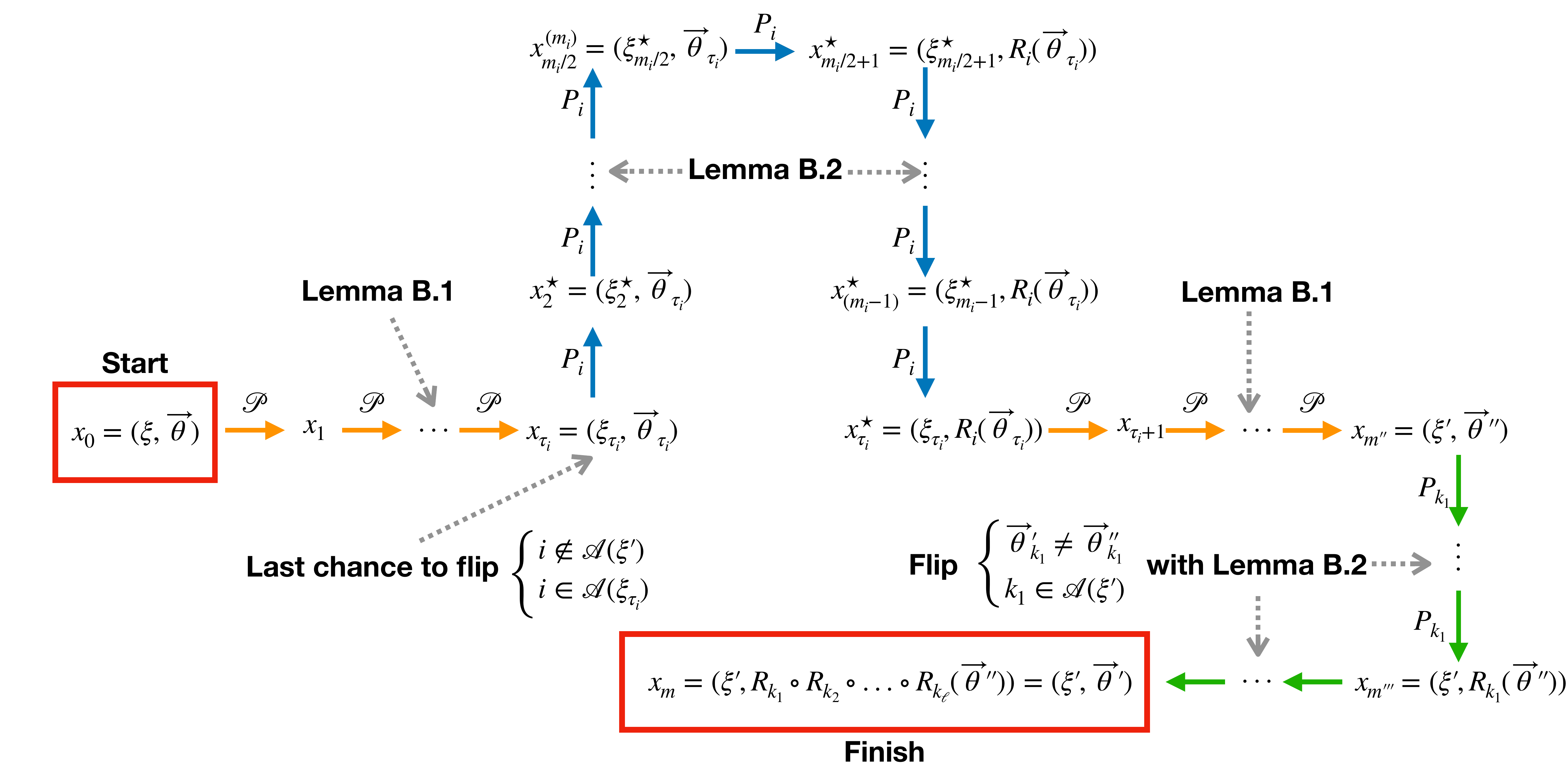}
\end{center}
\caption{A sketch of the proof of theorem 4.5, showing that there is a path with non-zero probability between any two points in the extended state space $(\dm, \vel)$ and $(\dm',\vel')$.}
\label{fig:proofoutline}
\end{figure}

\begin{proof}[Proof of \cref{thm:ergodic:general}]

As already mentioned, \cref{col:trans} ensures that for any extended state $(\dm,\velvec)$ and state $\dm'\in \dmDomain$, there is 
a sequence of states $(\dm_{0},\velvec_{0}),\cdots,  (\dm_{m},\velvec_{m''})$ with $(\dm_{0},\velvec_{0}) = (\dm,\velvec)$ and $\dm_{m''}=\dm'$ and $ \xtransc((\dm_{k},\velvec_{k}),(\dm_{k+1},\velvec_{k+1}))>0$ for all $k=0,\dots,m''-1$. Looking at  \cref{fig:proofoutline}, this central path which will modified is represented by the  sequence $\dm_1$ to $\dm_{m''}$ which runs horizontal across the center of the diagram (orange arrows).

In order to show \eqref{eq:xtransc:irred}, we need to show that this sequence can be modified/extended to a sequence where the final velocity coincides with the velocity vector $\velvec'$ specified in \eqref{eq:xtransc:irred}. First notice that if $\velvec''$ and $\velvec'$ only differ in components which are active in $\dm'$, then the existence of such a modified sequence directly follows from \cref{col:flippable:2}. One simply adds loops (green arrows in \cref{fig:proofoutline}) starting from $\dm'$ and the current $\velvec$ and returning with a particular component of the momentum's sign flipped to agree with $\velvec'$. In \cref{fig:proofoutline}, these excursions correspond extensions to the initial path, which ended at $\dm_{m''}$, which snake downward and then back to the left from $(\dm',\velvec'')$.

If $\velvec''$ and $\velvec'$ differ in a component $i$ which is activated along the path $(\dm_{0},\dots,\dm_{m''})$, then, we can again invoke  \cref{col:flippable:2} to insert a loop in the middle of the original sequence to flip the $i$th component of the momentum (blue arrows in \cref{fig:proofoutline}).
We consider the state $\dm_{\tau_{i}}$ in the sequence $(\dm_{0},\dots,\dm_{m''})$ at which the component $i$ is active for the last time, i.e., $\tau_{i}:= \max\{ j \in \{1,\dots,m''\} \mid i \in \activeset(\dm_{j}) \}$. After this state we insert a loop $(\dm^{\star}_{1},\velvec^{\star}_{1}),\dots,(\dm^{\star}_{m_{i}},\velvec^{\star}_{m_{i}})$, into the sequence which according to \cref{col:flippable:2} flips the sign of the $i$th velocity component, i.e., $\dm^{\star}_{m_{i}} = \dm_{\tau_{i}}$, $\velvec^{\star}_{m_{i}} = R_{i}(\velvec_{\tau_{i}})$. This produces a modified sequence of states 
\[
(\dm_{0},\velvec_{0}),\dots,(\dm_{\tau_{i}},\velvec_{\tau_{i}}),(\dm^{\star}_{1},\velvec^{\star}_{1}),\dots,(\dm^{\star}_{m_{i}},\velvec^{\star}_{m_{i}}), (\dm_{\tau_{i}+1},R_{i}(\velvec_{\tau_{i}+1})), \dots,   (\dm_{m''},R_{i}(\velvec_{m''})), 
\]
which  is attained with positive probability by the Markov chain (see \cref{fig:proofoutline}). Observe that, the transition probabilities between the remaining states $(\dm_{\tau_{i}+1},R_{i}(\velvec_{\tau_{i}+1})), \dots,   (\dm_{m},R_{i}(\velvec_{m}))$ is not affected by the sign change in the $i$th velocity component, since the $i$th component is by the choice of $\tau_{i}$ inactive for all this states. 
All that remains is to show that the  transition from the extended state $(\dm^{\star}_{m_{i}},\velvec^{\star}_{m_{i}}) = (\dm_{\tau_{i}}, R_{i}(\velvec_{\tau_{i}}))$ to $(\dm_{\tau_{i}+1},R_{i}(\velvec_{\tau_{i}+1}))$ does not use  kernel $P_{i}$; and hence, can not be effected by the  sign change in the $i$th velocity component from $\dm^{\star}_{m_{i}} = \dm_{\tau_{i}}$ (at the start of the loop) to $\velvec^{\star}_{m_{i}} = R_{i}(\velvec_{\tau_{i}})$ (at the end of the loop). \cref{lem:trans} ensures that $P_{i}( (\dm_{\tau_{i}}, R_{i}(\velvec_{\tau_{i}}), (\dm_{\tau_{i}+1},R_{i}(\velvec_{\tau_{i}+1}))=0$ because we know that $i \not\in  \activeset(\dm_{\tau_{i}+1})$ from the definition of $\tau_i$.

By repeating the described procedure for all remaining mismatched components which are activated somewhere along the original sequence of states $(\dm_{0},\dots,\dm_{m''})$, we obtain a sequence for which the final velocity coincides in all these components with $\velvec'$. Since \cref{col:trans} allows to choose the sequence $(\dm_{0},\dots,\dm_{m''})$ which connects the state $(\dm,\velvec)$ with $(\dm',\velvec')$, such that every component $\xtrans_{i}, ~i \in \{1,\dots,n\}$ is at least activated once on the way, this concludes the proof. 
\end{proof}

We now return to the proofs of \cref{col:trans}, \cref{col:flippable:2} and \cref{lem:trans} which we will give in reverse order.

\begin{proof}[Proof of \cref{lem:trans}]
Since mixed skew detailed balance holds for $\xtransc$ we have in particular 
\begin{equation}\label{eq:mskewed:app}
\xtrans_{i}((\dm,\velvec), (\dm,\velvec)) \target(\dm) \weight_{i}(\dm)=   \xtrans_{i}\left((\dm',R_{i}(\velvec)),(\dm,R_{i}(\velvec)) \right) \target(\dm') \weight_{i}(\dm'). 
\end{equation}
By the hypotheses of the Lemma we have $\weight_{i}(\dm')= 0$ since  $i \not \in \activeset(\dm')$. Thus, the right hand side of \eqref{eq:mskewed:app} certainly is $0$. Similarly, we have $\weight_{i}(\dm)>0$, and $\target(\dm)>0$ by assumption, and therefore we must have $P_{i}((\dm,\velvec),(\dm',\velvec))=0$ in order for the equality given in \eqref{eq:mskewed:app} to hold. 
\end{proof}

\begin{proof}[Proof of \cref{col:flippable:2}]
   First notice that since each kernel $\xtrans_{i}$ satisfies a modified skew detailed balance condition as detailed in \eqref{thm:MSDB}, we have that
\[
P_{i}((\dm,\velvec),(\dm',\velvec)) >0 \iff P_{i}((\dm',R_{i}(\velvec)), (\dm,R_{i}(\velvec)))>0\,.
\]
Thus it follows inductively that any sequence of states $(\dm_{0},\velvec), (\dm_{1},\velvec),\dots, (\dm_{m},\velvec_{m})$  which can be observed with positive probability when evolving according to $\xtrans_{i}$, can also with positive probability be ``walked back'' in reverse order as $(\dm_{m},R_{i}(\velvec)), \dots, (\dm_{1},R_{i}(\velvec)), (\dm_{0},R_{i}(\velvec))$ after flipping the sign of the $i$th momentum. Since $i \in  \activeset(\dm)$, \cref{as:prop:ir} guarantees that there is some path with positive $P_i$ probability  from $(\dm,\velvec)$ to a point at which $\velvec$ can be flipped to $R_{i}(\velvec)$. Without loss of generality, we can assume that $\velvec$ does not change until this last step. By ``walked back'' along this path we arrive at $(\dm, R_{i}(\velvec))$ as desired.
\end{proof}

\begin{proof}[Proof of \cref{col:trans}] From \cref{as:prop}, which states that $\prop$ is irreducible on $\Domain$, we know that there exists an integer $m$ and a path $\dm=\dm_0, \dots,\dm_m=\dm'$ with positive $\prop$ probability. By the construction of the $\xtransc$  chain from  $\prop$, we know that for any pair $\dm_i$ and $\dm_{i+1}$ along this path there exists some kernel $P_k$ and some choice of the $k$th momentum $\epsilon \in\{-1,1\}$ so that $P_k(  (\dm_i,\velvec^*), (\dm_{i+1}, \velvec^*)) >0$ iff $\vel^*_k=\epsilon$. If by chance we arrive at $\dm_i$ with the wrong sign in the $k$th momentum we can insert one of the loops constructed in  \cref{col:flippable:2} to change the sign of the $k$th momentum. By following this procedure inductively for each step in the original path, we can obtain the path needed to prove the first part of  \cref{col:trans}. 

Finally, we need to show that we can construct the path connecting $(\dm,\velvec)$ to $(\dm',\velvec')$ such  for each index $i \in \{1,\dots,n\}$ there is that along one state $(\dm_k, \velvec_k)$ such that $i \in \activeset(\dm_k)$. First notice that \cref{rem:activeSomewhere} guarantees  the existence of states $\dm^*_i \in \Domain$ with $i \in \activeset(\dm_i^*)$ for each  $i \in \{1,\dots,n\}$. The last result is obtained by prepending to the path constructed above a path which visits each of the $\dm^*_i$, $i \in \{1,\dots,n\}$, before heading onto   $(\dm',\velvec')$.
\end{proof}

\subsection{Proof of \cref{lem:flippable}}
We closeout this last appendix of the paper by proving  \cref{lem:flippable} which showed that, under  \cref{as:prop},  \cref{as:circuit}  and  \cref{as:prop:ir} are  equivalent.
\begin{proof}[Proof of \cref{lem:flippable}] 
Consider the Markov chain obtained by starting in $(\dm,\velvec)$ and evolving according to the memory kernel $\xtrans_{i}$.  
Moreover, assume without loss of generality $\vel_{i}=+1$, and denote by 
\[
\vertices_{i}^{+}(\dm) = \bigcup_{m \in \NN} {\rm supp}\, \xtrans_{i}^{m}((\dm,\velvec), (\,\cdot\,,\velvec)),
\]
the set of states which can be reached from $\dm$ in a finite number of steps when evolving according to the kernel $\xtrans_{i}$ without changing the sign of the velocity $\vel_{i}=+1$. 
This set of reachable vertices $\vertices_{i}^{+}(\dm)$ induces a subgraph of  $\graph_{i}^{+}$ which we denote by $\graph_{i}^{+}(\dm)=(\vertices_{i}^{+}(\dm),\edges_{i}^{+}(\dm))$ where $\edges_{i}^{+}(\dm):= \edges_{i}^{+}\cap \left (   \vertices_{i}^{+}(\dm)\times \vertices_{i}^{+}(\dm) \right )$.

Let \cref{as:circuit} be violated. That is for a certain $i\in \{1,\dots,n\}$ the corresponding graph $\graph_{i}^{+}$ contains a non-escapable circuit . If $\dm$ is contained in this circuit , then $\graph_{i}^{+}(\dm)$ coincides with this closed circuit  and in particular $\xtrans_{i}((\dm',\velvec),(\dm',R_{i}(\velvec))) = 0$ for every $\dm' \in \graph_{i}^{+}(\dm)$. Thus, it follows immediately that also \cref{as:prop:ir} is violated. This shows that \cref{as:prop:ir} implies \cref{as:circuit}. 

In order to show that \cref{as:circuit} together with \cref{as:prop} implies \cref{as:prop:ir} it suffices to show that if  \cref{as:circuit} and \cref{as:prop} hold, then there is always at least one state/node, say $\tilde{\dm}$, in the set of reachable vertices $\vertices_{i}^{+}(\dm)$ for which there is a positive probability of flipping the velocity component $\vel_{i}$ (in the sense that $\xtrans_{i}((\tilde{\dm}, \velvec), (\tilde{\dm}, R_{i}(\velvec)))>0$.) 
In order to show that such a state indeed always exists we consider the cases where  $\graph_{i}^{+}(\dm)$ either does or does not contain a non-escapable circuit  separately. 

If the graph $\graph_{i}^{+}(\dm)$ contains a non-escapable circuit, then, the existence of such a state $\tilde{\dm}$ is guaranteed by \cref{as:circuit}. 

If the graph $\graph_{i}^{+}(\dm)$  does not contain a non-escapable circuit , the existence of such a state $\tilde{\dm}$ can be easily shown using the fact that the probability measure $\xtarget_{i}(\dm',\velvec') \propto \weight_{i}(\dm')\target(\dm')$ is invariant under $\xtrans_{i}$ (see \cref{rem:sqewMix}): there is at least one state $\dm^{\star} \in  \edges_{i}^{+}(\dm)$ which is not part of a circuit  (otherwise $\graph_{i}^{+}(\dm)$ would be a non-escapable circuit). If $\xtrans_{i}((\dm', \velvec), (\dm', R_{i}(\velvec))) = 0$ for all $\dm' \in\vertices_{i}^{+}(\dm) \supset \vertices_{i}^{+}(\dm^{\star})$, then this state is transient, which is in direct contradiction to $\xtarget_{i}(\dm^{\star},\velvec)>0$. Consequently, there must be at least one state in $\vertices_{i}^{+}(\dm)$ for which the probability of flipping the $i$th velocity component is positive. 
\end{proof}

\end{document}